\documentclass[12pt,a4paper]{amsart}

\usepackage[utf8]{inputenc}
\usepackage[english]{babel}
\usepackage[T1]{fontenc}
\usepackage{amsmath}
\usepackage{mathtools}
\usepackage{amsfonts}
\usepackage{amssymb}
\usepackage{setspace}
\usepackage{amsfonts, dsfont}
\usepackage{mathrsfs, enumerate, csquotes, color}
\usepackage{xcolor}
\definecolor{vertfonce}{rgb}{0.20, 0.46, 0.25}
\definecolor{rougefonce}{rgb}{0.64, 0.09, 0.20}
\usepackage[breaklinks=true,
            colorlinks=true,
            linkcolor=rougefonce,
            citecolor=vertfonce]{hyperref}

\newcommand{\h}{\hbar}
\renewcommand{\O}{\mathscr{O}}
\newcommand{\abs}[1]{\left|#1\right|}

\author{R. Fahs}
\email{rayan.fahs@univ-angers.fr}
\address[R. Fahs]{Univ Angers, CNRS, LAREMA, SFR MATHSTIC, F-49000 Angers, France}

\author{L. Le Treust}
\email{loic.le-treust@univ-amu.fr}
\address[L. Le Treust]{Aix Marseille Univ, CNRS, Centrale Marseille, I2M, Marseille, France}

\author{N. Raymond}
\email{nicolas.raymond@univ-angers.fr}
\address[N. Raymond]{Univ Angers, CNRS, LAREMA, Institut Universitaire de France, SFR MATHSTIC, F-49000 Angers, France}

\author{S. V\~u Ng\d{o}c}
\email{san.vu-ngoc@univ-rennes1.fr}
\address[S. V\~u Ng\d{o}c]{Univ Rennes, CNRS, IRMAR - UMR 6625, F-35000 Rennes, France}

\title[]{Boundary states of the Robin magnetic Laplacian}

\makeatletter

\@addtoreset{equation}{section}
\makeatother

\usepackage{tikz}
\usetikzlibrary{shapes.misc}
\tikzset{cross/.style={cross out}, minimum size=1pt, draw=black, inner sep =0pt, outer sep=0pt, cross/.default={1pt}}

\theoremstyle{plain}
\newtheorem{theorem}{Theorem}[section]
\newtheorem{lemma}[theorem]{Lemma}
\newtheorem{corollary}[theorem]{Corollary}
\newtheorem{proposition}[theorem]{Proposition}

\theoremstyle{definition}
\newtheorem{remark}[theorem]{Remark}

\newtheorem{definition}[theorem]{Definition}
\newtheorem{assumption}[theorem]{Assumption}

\newcommand{\R}{\mathbf{R}}
\newcommand{\T}{\mathbf{T}}
\newcommand{\C}{\mathbf{C}}
\newcommand{\N}{\mathbf{N}}
\newcommand{\Z}{\mathbf{Z}}

\renewcommand{\Re}{\mathrm{Re}}
\renewcommand{\Im}{\mathrm{Im}}

\renewcommand{\leq}{\leqslant}	\renewcommand{\geq}{\geqslant}

\newcommand{\dd}{\mathrm{d}}

\begin{document}

\maketitle

\begin{abstract}
  This article tackles the spectral analysis of the Robin Laplacian on
  a smooth bounded two-dimensional domain in the presence of a
  constant magnetic field. In the semiclassical limit, a uniform
  description of the spectrum located between the Landau levels is
  obtained. The corresponding eigenfunctions, called edge states, are
  exponentially localized near the boundary. By means of a microlocal
  dimensional reduction, our unifying approach allows on the one hand
  to derive a very precise Weyl law and a proof of quantum magnetic
  oscillations for excited states, and on the other hand to refine
  simultaneously old results about the low-lying eigenvalues in the
  Robin case and recent ones about edge states in the Dirichlet case.
\end{abstract}

% \tableofcontents

\section{Motivations and results}

\subsection{About the magnetic Robin Laplacian}
We want to describe the spectrum of the semiclassical magnetic
Laplacian $\mathscr{L}_h=(-i h\nabla-\mathbf{A})^2$ on a smooth,
bounded, and simply connected open Euclidean domain
$\Omega\subset\R^2$, with boundary conditions of Robin type.  The
vector potential $\mathbf{A} : \overline{\Omega}\to\R^2$ is supposed
to be smooth and generating a constant magnetic field of intensity
$1$:
\[
  \partial_1 A_2-\partial_2A_1=1\,.
\]
The magnetic Robin boundary conditions are enforced by defining the
operator $\mathscr{L}_h=\mathscr{L}_{h,A,\gamma}$ to be the
selfadjoint operator associated with the quadratic form defined for
all $\psi\in H^1(\Omega)$ by:
\begin{equation}
  \mathcal{Q}_{h,\mathbf{A}}(\psi) = \int_{\Omega}|(-i
  h\nabla-\mathbf{A})\psi|^2\mathrm{d}x+\gamma
  h^{\frac32}\int_{\partial\Omega}|\psi|^2\mathrm{d}s\,,\label{eq:Q}
\end{equation}
where $\gamma\in\R\cup\{+\infty\}$, and $\mathrm{d}s$ is the length
measure of the boundary induced by the Euclidean metric. By
convention, $\gamma=+\infty$ corresponds to the Dirichlet boundary
condition $\psi\in H^1_0(\Omega)$. In the whole paper, our estimates
will be uniform when $\gamma\in[-\gamma_0,+\infty]$ for an arbitrary
fixed $\gamma_0>0$. When $\gamma\in\R$, the domain of $\mathscr{L}_h$
is given by
\[
  \mathrm{Dom}(\mathscr{L}_h)=\{\psi\in H^1(\Omega) : (-i
  h\nabla-\mathbf{A})^2\psi\in L^2(\Omega)\,, -i h\mathbf{n}\cdot(-i
  h\nabla-\mathbf{A})\psi=\gamma h^{\frac32}\psi\,\mbox{ on
  }\partial\Omega\}\,,
\]
where $\mathbf{n}$ is the outward pointing normal to the
boundary. Note that a change of gauge can be used to ensure that
$\mathbf{A}\cdot\mathbf{n}=0$. In this case, the magnetic Robin
condition becomes a usual Robin condition:
\begin{equation}\label{eq.boundary}
  -\mathbf{n}\cdot\nabla\psi=\gamma h^{-\frac12}\psi\,.
\end{equation}
We would like to establish accurate spectral asymptotics for
$\mathscr{L}_h$ in regimes where the magnetic field plays a major
role, competing with the Robin condition (this is the origin, as we
will see, of the factor $h^{\frac32}$ in the Robin condition). Until
now, very accurate results are available only for the Neumann magnetic
Laplacian (when $\gamma=0$). In this case, the lowest eigenvalues have
been analyzed in detail in \cite{FH06} and uniform estimates have been
recently established in \cite{BHR21} where a purely magnetic
tunnelling effect formula has been proved. When $\gamma\neq 0$, the
only known results go back to the works by Kachmar, see \cite{K06},
where only the smallest eigenvalue has been estimated. In all these
situations (except when $\gamma=+\infty$, where the first eigenvalue
is asymptotic to $h$ times the magnetic intensity --- here, 1), one
can show that the first eigenvalue becomes smaller than $h$ as soon as
$h$ is small enough. This energy bound is usually associated with a
localization behavior near the boundary of the eigenfunctions, which
can be quantified by semiclassical Agmon estimates.

By a simple scaling, the semiclassical limit $h\to 0$ translates
into a quantum regime where the intensity of the magnetic field tends
to infinity. In the physics literature of thin conductors or electron
gases (approximated by 2D domains) subject to a strong external
magnetic field, it is well known that the presence of a boundary (or,
more generally, of an abruptly changing magnetic field along a curve)
generates a current along the boundary due to the presence of
``bouncing modes'' classically localized at a distance $\sqrt E/B$ to
the boundary ($E$ is the kinetic energy and $B$ is the magnetic
intensity: in this work $B=1$), see for
instance~\cite{Halperin82}. These so-called ``edge states'' or
``boundary states'' exist as soon as the Fermi level of the conductor
lies strictly in between two consecutive Landau levels, and produce
ballistic dynamics along the boundary. If the boundary
$\partial\Omega$ is compact, this dynamics is quantized and produces
new discrete energy levels. These are precisely the eigenvalues that
we wish to describe in this work.

Heuristically, the localization near $\partial\Omega$ is often
explained by the classical bouncing modes alluded to above, but it is
also easy to understand from a quantum perspective.  Indeed, if we
forget the boundary condition, $\mathscr{L}_h$ acts as the magnetic
Laplacian with constant magnetic field, on the Euclidean plane,
$\mathscr{L}^{\R^2}_{h,\mathbf{A}}$. The spectrum of this so-called
``bulk'' operator is well-known and made of the famous Landau levels
$\{(2n-1)h\,, n\geq 1\}$, which are infinitely degenerate
eigenvalues. This suggests that, if one considers potential
eigenvalues of $\mathscr{L}_h$ in a window of the form $I_h=[h a,h b]$
with $2n-1<a<b<2n+1$ for some integer $n\geq 0$ (for $n=0$, we take
$a=-\infty$), they cannot correspond to any bulk state, and hence the
corresponding eigenfunctions should be localized near the boundary.
This phenomenon has interesting physical applications; a famous one is
the quantum Hall effect, when the domain is not simply connected,
which expresses the collective effects of several boundaries on the
total net current. Another application is the confinement of particles
in small domains, or ``quantum dots'' (sometimes called ``anti-dots''
because one takes $B=0$ inside the domain, and $B=1$ outside),
see~\cite{Reijniers-Peeters-Matulis99,Lee_2004}.

On the mathematics side, the existence of edge currents in a
half-plane with Dirichlet boundary condition was shown
in~\cite{de_Bievre-Pule99}. In a compact setting, the eigenfunction
localization at the boundary has been observed (again in the Dirichlet
case $\gamma=+\infty$, which is usually chosen in physics) in
\cite{GV21}, which was one of our motivations for this work. The
methods of~\cite{GV21} lead to a description of the spectrum in a thin
spectral window, see \cite[Corollary 2.7]{GV21}. However the
\emph{exponential} decay was not established. In fact, as we will see,
this decay does not  follow from the usual Agmon estimates, but from a strategy \textit{à la} Combes-Thomas (see the original article \cite{CT73} or the review \cite{Hislop00}).

In this article we treat the general case
$\gamma\in\R\cup \{+\infty\}$. This corresponds, physically, to a
domain $\Omega$ coated with a very thin layer of a different material
(see for instance~\cite{Bendali-Lemrabet1996}). Since $\Omega$ is
bounded, the spectrum in $I_h$ is always discrete and a first rough
estimate shows that the number of eigenvalues lying in $I_h$, denoted
by $N(\mathscr{L}_h,I_h)$, satisfies
\begin{equation}\label{eq.roughNLhIh}
  N(\mathscr{L}_h,I_h)\leq Ch^{-2}\,,
\end{equation}
for some $C>0$ and all $h>0$ small enough (see Appendix
\ref{sec.roughWeyl} where we recall the origin of this estimate).  Our
goal is to obtain a very precise description, in the semiclassical
regime, of the spectral elements corresponding to the interval $I_h$,
much more accurate than~\eqref{eq.roughNLhIh}. This includes the
localization behavior near $\partial\Omega$ of the corresponding
eigenfunctions. For instance, when $\gamma\in\R$, a consequence of our
main result Theorem~\ref{thm.main} is the appearance of a quite
interesting phenomenon: for a given (low) energy, one can have
boundary quasimodes corresponding to classical currents flowing in
opposite directions, leading to magnetic oscillations of eigenvalues,
see Theorem~\ref{theo:oscillate}.

This work is also an opportunity to revisit the Neumann case analyzed
in \cite{HM01, FH06} (see also \cite{BHR21}) by establishing more
uniform asymptotic expansions, with slightly more general boundary
conditions.

\subsection{De Gennes operator with Robin condition}
\label{sec:de-gennes-operator}

Our results will be expressed in terms of the eigenvalues of the de
Gennes operator with Robin boundary condition. This operator, which
appears naturally in the study of boundary induced magnetic
effects~\cite{StJames-deGennes63,FH10}, is a differential operator of
order two depending on the real parameters $\gamma$ and $\sigma$ and
acting as
\[
  H[\gamma,\sigma] = -\dfrac{\dd ^2}{\dd t ^{2}}+(t-\sigma)^2\,,
\]
on the domain
\[
  \mathrm{Dom}(H[\gamma, \sigma])=\left\{u \in
    B^{1}\left(\R_{+}\right) : \Big(-\dfrac{\dd ^2}{\dd t
      ^{2}}+(t-\sigma)^2\Big)u \in L^2(\R_{+}),\,\,
    u^{\prime}(0)=\gamma\, u(0)\right\}\,,
\]
where
\[
  B^{1}\left(\R_{+}\right) = \{ u\in
  H^{1}\left(\R_{+}\right)\,:\;[t\mapsto tu(t)]\in
  L^{2}\left(\R_{+}\right)\}\,.
\]
It is well-known that $H[\gamma,\sigma]$ is a self-adjoint elliptic
operator with compact resolvent. Its spectrum can be written as a
non-decreasing sequence of eigenvalues
$( \mu_n(\gamma,\sigma))_{n\geq 1}$ (which are all simple due to the
Cauchy-Lipschitz theorem).
We denote by $u_n^{[\gamma,\sigma]}$ the normalized sequence of the
corresponding eigenfunctions (with $u_n^{[\gamma,\sigma]}(0)>0$).  We
let
\[
  \Theta^{[n-1]}(\gamma):= \inf_{\sigma \in \R}
  \mu_{n}(\gamma,\sigma).
\]
The index $n-1$ is compatible with the notation used in the case of
the de Gennes operator (case when $\gamma=0$), see \cite[Section
3.2]{FH10}.  The family $(H[\gamma,\sigma])_{(\gamma,\sigma)\in\R^2}$
is analytic of type (B) (in the sense of Kato, see \cite[Chapter VII,
\S 4]{Kato}), \emph{i.e.}, the form domain does not depend on the
parameters and the sesquilinear form is analytic as a function of
$\gamma$ or $\sigma$. By convention, we denote by $H[+\infty,\sigma]$
(\emph{i.e.} we let $\gamma=+\infty$) the corresponding operator with
Dirichlet boundary condition $u(0)=0$.

The following proposition gathers the main properties of the functions
$\mu_{n}(\gamma,\cdot)$ (which are usually called \emph{dispersion
  curves}) that will be used in this article. Most of them have been
established in \cite{K06} (see also \cite{K07},
and~\cite{de_Bievre-Pule99} in the Dirichlet case).
\begin{proposition}\label{prop.dispersion}
  Let us fix $n\geq 1$. When $\gamma\in\R$, the function
  $\mu_{n}(\gamma,\cdot)$ is analytic and
  \begin{equation}
    \lim_{\sigma\to-\infty}\mu_n(\gamma,\sigma)=+\infty\,,\quad
    \lim_{\sigma\to+\infty}\mu_{n}(\gamma,\sigma)=2n-1\,.
    \label{equ:limites_sigma}
  \end{equation}
  Moreover, $\mu_n(\gamma,\cdot)$ has a unique minimum attained at
  $\sigma=\xi_{n-1}(\gamma)$, but not attained at infinity. This
  minimum is non-degenerate.  The function $\mu_n(\gamma,\cdot)$ is
  decreasing on $(-\infty,\xi_{n-1}(\gamma))$ and increasing on
  $(\xi_{n-1}(\gamma),+\infty)$. In addition, we have, for all
  $n\geq 2$,
  \begin{equation}\label{eq.lubtheta}
    2n-3<\Theta^{[n-1]}(\gamma)<2n-1\,.
  \end{equation}
  When $\gamma=+\infty$, that is when the Robin condition is replaced
  by the Dirichlet condition, $\mu_n(+\infty,\cdot)$ is still smooth,
  but now decreasing from $+\infty$ to $2n-1$.
\end{proposition}
The non-degeneracy of the minimum of $\mu_n(\gamma,\cdot)$ for
$\gamma\in\R$ is obtained by adapting the Dauge-Helffer formula,
see~\cite{K06} for the case $n=1$, which gives:
\begin{equation}
  {\partial^2_\sigma}\mu_n(\gamma,\sigma)_{| \sigma
    = \xi_{n-1}(\gamma)} = 2 \xi_{n-1}(\gamma)
  \abs{u_{n}^{[\gamma,\sigma]}(0)}^2 \,.
  \label{equ:partial2sigma(mu)}
\end{equation}
The lower bound in \eqref{eq.lubtheta} will be established in Appendix
\ref{sec.deGennes}. This proposition has the following elementary but
important consequences for our analysis, which are illustrated in
Figure~\ref{fig:dispersion}.

\begin{corollary}\label{coro:dispersion}
  Let $\gamma\in\R\cup \{+\infty\}$ be fixed. Let $\Theta$ be the set
  of all critical values of the functions $\mu_n$: we have
  \[
    \Theta = \{\Theta^{[n-1]}(\gamma), \quad n\geq 1\}\,.
  \]
  Let $\Lambda$ be the set of limit points of the functions $\mu_n$ at
  infinity:
  \[
    \Lambda:= \{ 2n-1, \quad n\geq 1\}\,.
  \]
  Let $[a,b]\subset \R$ be an interval disjoint from $\Lambda$. Let
  either $n=1$ if $a<1$ or $n\geq 2$ be such that
  $[a,b]\subset (2n-3, 2n-1)$. (In the case $n=1$ we allow
  $a=-\infty$.)  It follows from~\eqref{equ:limites_sigma} that for
  any integer $k\geq 1$, $\mu_k^{-1}([a,b])$ is compact.

  Let $p(k)$ be the number of connected components of
  $\mu_k^{-1}([a,b])$: we have
  \[
    \begin{cases}
      p(k)=1 & \textup{if } 1 \leq k < n\\
      p(n) = 0 & \textup{if } \gamma=+\infty\\
      p(n) = 1 & \textup{if } \gamma\in\R \textup{ and } \Theta^{[n-1]}\in[a,b]\\
      p(n) = 2  & \textup{if } \gamma\in\R \textup{ and } \Theta^{[n-1]} < a\\
      p(n) = 0  & \textup{if } \gamma\in\R \textup{ and } b < \Theta^{[n-1]}\\
      p(k)=0 & \textup{if } k > n\,.
    \end{cases}
  \]
  Therefore, when $\gamma\in\R$,
  \begin{equation}
    N(\gamma,a,b) := \# \{k\geq 1 :
    \mu_k(\gamma,\cdot)^{-1}([a,b])\neq\emptyset\}=
    \begin{cases}
      n &\mbox{ if } b\geq\Theta^{[n-1]}(\gamma)\\
      n-1&\mbox{otherwise}
    \end{cases}\,,
    \label{equ:N}
  \end{equation}
  and if $\gamma=+\infty$ (Dirichlet case) then $\mu_1(+\infty,\cdot)$
  does not take any value in $(-\infty,1)$, and
  $N(\gamma,a,b) = n - 1$.
\end{corollary}
  
% For $n\geq 1$, let us consider an interval of the form
% $[a,b]\subset(2n-3,2n-1)$ (if $n\geq 2$) or of the form
% $(-\infty,b]\subset(-\infty,1)$. Moreover, has the following
% important consequence: see
\begin{figure}[h]
  \centering \includegraphics[width=\linewidth]{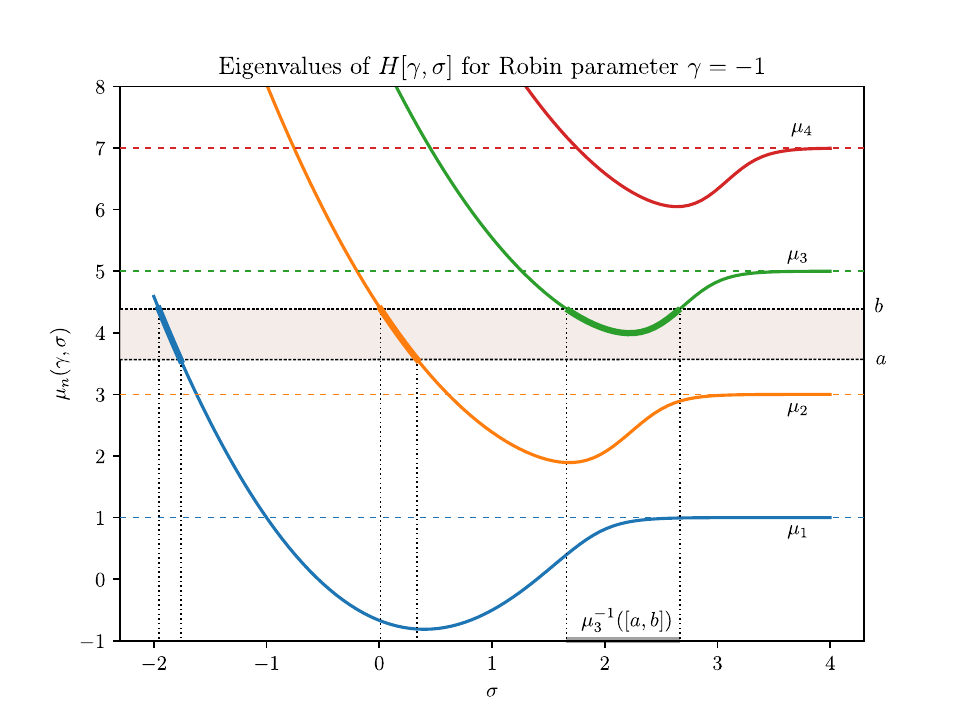}
  \caption{This figure shows the dispersion curves
    $\mu_k(\gamma,\cdot)$, for $\gamma=-1$.  We visualise the preimage
    of a given interval $[a,b]\subset (2n-3, 2n-1)$ with $n=3$. The
    curves are obtained by a standard finite difference numerical
    scheme.}
  \label{fig:dispersion}
\end{figure}
From now on, we denote by $N(\gamma,a,b)=N$ this cardinal.
\begin{assumption}\label{rem.abh}
  In the following, $a$ and $b$ are allowed to depend on $h$, as soon
  as they stay in an $h$-independent compact interval inside
  $(2n-3, 2n-1)$.
\end{assumption}
With this picture in mind, for each $k\geq 1$, we may construct a
smooth function $\overset{\circ}{\mu}_k$, bounded with all its
derivatives, which coincide with $\mu_k$ in a neighborhood of
$\mu_k^{-1}([a,b])$. Indeed, let $\Xi_0: \R \to \R$ be a smooth,
bounded with all its derivatives, and increasing function such that
for all $k\in\{1,\ldots,N\}$,
$\mu_k(\gamma,\Xi_0(\sigma))=\mu_k(\gamma,\sigma)$ in a neighborhood
of $\mu_k^{-1}([a,b])$ and $\mu_k\circ\Xi_0$ takes its values in
$(-\infty,a)\cup(b,+\infty)$ away from it. We consider
\begin{equation}
  \overset{\circ}{\mu}_k:=\mu_k (\gamma,\Xi_0(\cdot))\,,
  \label{equ:circle}
\end{equation}
where we omit
the reference to the parameter $\gamma$ to lighten the notation.  In
the following, we will more generally denote by
$\overset{\circ}{\varphi}$, the function $\varphi$ after $\Xi_0$.

\subsection{Results}
Let us now describe the main results of our article, which will be
expressed in terms of pseudo-differential operators in one dimension.

\subsubsection{A pseudo-differential framework}
The bounded functions $\overset{\circ}{\mu}_k$ will be convenient to
state our main theorem, which involves
$h^{\frac12}$-pseudo-differential operators with symbols in the usual
class $S_{\R^2}(1)$ given by
\[
  S_{\R^2}(1)=\{a\in\mathscr{C}^{\infty}(\R^2_{s,\sigma}) :
  \forall\alpha\in\N^2\,,\exists C_\alpha>0 : |\partial^\alpha a|\leq
  C_\alpha\}\,.
\]
As we said before, the eigenfunctions of $\mathscr{L}_h$ will be
localized near the boundary of $\Omega$, which is a closed smooth
curve with length $2L$. Our main result describes their distribution
with the help of an $h^{\frac12}$-pseudo-differential operator on the
boundary (see for instance \cite[Section 4.1]{FHKR22} where similar
considerations have been done in the context of discontinuous magnetic
fields). Let us denote $\hbar=h^{\frac12}$. We recall that the Weyl
quantization of a symbol $p$ is given by the formula:
\begin{equation}\label{eq.Weylquant}
  (\mathrm{Op}^{\mathrm{W}}_\hbar p) \psi(x) =
  \frac{1}{2\pi\hbar}\int_{\R^2}e^{i(x-y)\eta/\hbar}
  p\left(\frac{x+y}{2},\eta\right)\psi(y)\mathrm{d}y\mathrm{d}\eta\,,
  \quad \forall\psi\in\mathscr{S}(\R)\,,
\end{equation}
and that this formula defines a bounded operator form $L^2(\R)$ to
$L^2(\R)$ if $p\in S_{\R^2}(1)$, by the Calderón-Vaillancourt
theorem. To shorten the notation, we will sometimes write
$p^{\mathrm{W}}$ instead of $\mathrm{Op}^{\mathrm{W}}_\hbar p$.

Let $\T_{2L} = \R/2L\Z$, and $L^2(\T_{2L})$ be the subset of
$L^2_{\mathrm{loc}}(\R)$ of $2L$-periodic functions, equipped with the
usual $L^2$ norm on $[0,2L]$. When $p\in S_{\T_{2L}\times \R}(1)$,
\emph{i.e.} $p\in S_{\R^2}(1)$ and is $2L$-periodic in its first
variable $s$, then for any $\theta\in\R$, the operator given in
\eqref{eq.Weylquant} induces a bounded operator from
$e^{i\theta\cdot}L^2(\T_{2L})$ to $e^{i\theta\cdot}L^2(\T_{2L}) $ ---
we denote by $e^{i\theta\cdot}$ the function $x\mapsto e^{i\theta x}$.

\begin{remark}\label{rem.period}
  The space $e^{i\theta\cdot}L^2(\T_{2L})$ only depends on the class of
  $\theta$ modulo $\frac{\pi}{L}$; it is equal to the subspace of
  functions in $L^2_{\mathrm{loc}}(\R)$ equipped with the Floquet
  boundary condition $\psi(x+2L) = e^{i2L\theta}\psi(x)$.  The
  operator $\mathrm{Op}^{\mathrm{W}}_\hbar p$ acting on
  $e^{i\theta\cdot}L^2(\T_{2L})$ is unitarily equivalent to
  $\mathrm{Op}^{\mathrm{W}}_\hbar p(x,\eta+\hbar\theta)$ acting on
  $L^2(\T_{2L})$.
\end{remark}

\subsubsection{Main theorem}
Since our main result describes the spectrum of $\mathscr{L}_h$
``modulo $\O(h^\infty)$'', we need to make this notion precise.
\begin{definition}\label{defi:coincide}
  In this article, we will say that the spectra of two self-adjoint
  operators $T_1$ and $T_2$ depending on $h$ coincide in $I_h$ modulo
  $\mathscr{O}(h^\alpha)$, $\alpha\in\R\sqcup\{+\infty\}$, when there
  exists $C,h_0>0$ such that, for all $h\in(0,h_0)$,
  \begin{enumerate}[---]
  \item $T_1$ and $T_2$ have discrete spectrum in
    $I_h+[-Ch^{\alpha},Ch^\alpha]$,
  \item for all interval $J_h\subset I_h$ we can find an interval
    $ K_{h}$ such that $J_h\subset K_{h}$ with
    $\mathrm{d}_{H}(K_{h},J_h)\leq Ch^\alpha$ and
    \[
      \mathrm{rank}\,\mathds{1}_{ J_h}(T_1) \leq
      \mathrm{rank}\,\mathds{1}_{K_{h}}(T_{2})
      \,,\quad\mathrm{rank}\,\mathds{1}_{ J_h}(T_2) \leq
      \mathrm{rank}\,\mathds{1}_{K_{h}}(T_{1}) \,,
      \] where $\mathrm{d}_{H}$ denotes the Hausdorff distance:
      \[
        \mathrm{d}_{H}(A,B)=\sup_{(a,b)\in A\times
          B}\max(d(a,B),d(b,A))\,.
      \]
    \end{enumerate}
   This definition translates to discrete subsets of $\R$ as
  follows: for each discrete subset $S\subset \R$, we associate the
  sum of Dirac masses $\delta_S:=\sum_{s\in S} \delta_s$, and consider
  the corresponding self-adjoint operator whose spectral measure is
  $\delta_S$. Then we say that two discrete subsets $A_1$ and $A_2$
  coincide modulo $\O(h^\alpha)$ when the spectra of the corresponding
  operators coincide modulo $\O(h^\alpha)$ in the above sense. In
  order to deal with multiplicities, we will, by convention, associate
  with the \emph{disjoint union} $S\sqcup S'$ the operator
  corresponding to the spectral measure $\delta_S+\delta_{S'}$.
\end{definition}

\begin{remark}
  \begin{enumerate}[---]
  \item The relation "the spectra of $T_1$ and $T_2$ coincide in $I_h$
    modulo $\mathscr{O}(h^\alpha)$'' is an equivalence relation. It is
    obviously symmetric and reflexive (taking $K_{h} = J_h$). The
    transitivity follows from the triangle inequality for $d_H$.

  \item If the spectra of $T_1$ and $T_2$ coincide in $I_h$ modulo
    $\mathscr{O}(h^\alpha)$, then, for all
    $\widetilde{I}_h\subset {I}_h$, the spectra of $T_1$ and $T_2$
    coincide in $\widetilde{I}_h$ modulo $\mathscr{O}(h^\alpha)$.

  \item If the spectra of $T_1$ and $T_2$ coincide in $I_h$ modulo
    $\mathscr{O}(h^\alpha)$, we have
    \[
      \mathrm{d}_{H}(\mathrm{sp}(T_1)\cap I_h,\mathrm{sp}(T_2)\cap
      I_h)=\mathscr{O}(h^\alpha)\,.
    \]
  \item If the endpoints of the interval $I_h$ stay away from an
    $h^\beta$-neighborhood of the spectrum, with $\beta<\alpha$, then
    for $h$ small enough $T_1$ and $T_2$ have exactly the same number
    of eigenvalues inside $I_h$, counted with multiplicities.
  \item The notion described in Definition \ref{defi:coincide} already
    appears under various forms in the literature (see, for instance,
    the view point in \cite[Section 1]{HS84} and \cite[\S
    4]{Helffer88}).
  \end{enumerate}
\end{remark}

We can now state our main result, where we use, among others, the
eigenvalues $\mu_k(\gamma,\sigma)$ and eigenfunctions
$u_k^{[\gamma,\sigma]}$ of the de Gennes operator (Section
\ref{sec:de-gennes-operator}), the integer $N$ defined
in~\eqref{equ:N}, and the notation introduced in~\eqref{equ:circle}.

\begin{theorem}\label{thm.main}
  Under Assumption \ref{rem.abh}, the spectrum of $\mathscr{L}_h$ in
  $[ha,hb]$ coincides with that of $h\mathfrak{M}_h$
  modulo $\mathscr{O}(h^2)$, where
\[
  \mathfrak{M}_h := \begin{bmatrix}
    m_{1}^{\textup{W}} & 0 & \cdots & 0 \\
    0 & m_{2}^{\textup{W}} &  & \vdots \\
    \vdots &  &  \ddots & 0 \\
    0 & \cdots & 0 & m_{N}^{\textup{W}}
 \end{bmatrix}
\]
 is a bounded operator acting diagonally on
  $e^{i\theta(h)\cdot}L^2(\T_{2L})^N$.  Here
  \[\theta(h)=\frac{|\Omega|}{|\partial\Omega|h}\,,\]
  and each $m_k^{\textup{W}}$ is an $h^{\frac12}$-pseudodifferential
  operator with symbol in $S_{\T_{2L}\times \R}(1)$. Let us denote by
  $(s,\sigma)$ the (canonical) variables in $\T_{2L}\times \R$. Then
  we have:
  \begin{itemize}
  \item the principal symbol of $m_k^{\textup{W}}$ is
    $\overset{\circ}{\mu_k}(\sigma)$;
  \item its subprincipal symbol is
    $-\kappa(s){\overset{\circ}C_k}(\sigma)$ with
    \begin{equation}\label{eq.Cj}
      C_k(\sigma)=\big\langle \left( (\tau-\sigma)
        \tau^2-\partial_\tau-2\tau(\sigma-\tau)^2\right)u^{[\gamma,\sigma]}_k(\tau),\;
      u^{[\gamma,\sigma]}_k(\tau)\big\rangle_{L^2(\R_+)}\,,
    \end{equation}
    and $\kappa(s)$ is the curvature of the boundary at the point of
    curvilinear abscissa $s$.
  \end{itemize}
\end{theorem}

\begin{remark}\label{rem16}
  One can check that, for all $k\geq 1$, $C_k(\xi_{k-1}(\gamma))$ has
  the same sign as $\gamma_0^{[k-1]}-\gamma$, see Proposition
  \ref{prop.sgnCj} where the threshold $\gamma_0^{[k-1]}$ is
  discussed. Proposition \ref{prop.sgnCj} also corrects a mistake in
  \cite[Lemma II.3 \& (2.24)]{K06}, where it is stated that
  $C_1(\xi_{0}(\gamma))$ is always positive.
\end{remark}
It is important to notice that Theorem \ref{thm.main} is actually a
diagonalization result since it reduces the spectral analysis of
$\mathscr{L}_h$ to that of a family of pseudo-differential operators
in one dimension: the spectrum of $\mathfrak{M}_h$ is the
superposition (counting multiplicities) of the spectra of the
$m_k^{\textup{W}}$, $k=1,\dots, N$.  As it turns out, the spectrum of
each of these pseudo-differential operators can be completely
described using (refinements of) old and new results in the
literature. Indeed, notice that the principal symbols
$\overset{\circ}{\mu_k}$ have a special feature: they depend only on
the frequency variable $\sigma$, and, as functions of $\sigma$, they
have at most a unique critical point, which is a nondegenerate minimum
(Proposition~\ref{prop.dispersion}). Hence, from a microlocal
viewpoint, only two situations must be considered. Let $E\in[a,b]$,
either $E$ is a regular value of $\overset{\circ}{\mu_k}$ (or
$\mu_k(\gamma,\cdot)$, equivalently), and then the well-known
Bohr-Sommerfeld rules apply, or $E$ is a critical value of
$\overset{\circ}{\mu_k}$, in which case the Hamiltonian
$(s,\sigma) \mapsto \overset{\circ}{\mu_k}(\sigma)$ admits a
transversally non-degenerate minimum on a circle, and the recent
study~\cite{D-VN21} of folded quantum action variables applies.

\subsubsection{Eigenvalues in a regular spectral window}

Our first application concerns the case where the interval $[a, b]$
consists of regular values of all $\mu_k$.  We will use the following
well-known spectral result, an extension to all orders of the
Bohr-Sommerfeld rules (see, for instance, \cite{Rozenblum75, HR82,
  HR84, Duistermaat74}), which we prove in
Section~\ref{sec:proof-BSreg}.

\begin{proposition}\label{prop.pseudo1D}
  Consider an $\h$-pseudo-differential operator
  $\mathsf{P}_\h\in\mathrm{Op}^{\mathrm{W}}_\hbar(S_{\R^2}(1))$ with
  symbol $2L$-periodic with respect to $s$ and with principal symbol
  $(s,\sigma)\mapsto\mu(\sigma)$ and subprincipal symbol
  $(s,\sigma)\mapsto -\kappa(s)C(\sigma)$. We consider its realization
  on $e^{is\theta(\hbar^2)}L^2(\T_{2L})$. Let $E$ be a regular value
  of $\mu$ for which $\mu^{-1}(E)$ is a finite set of points
  $\sigma^E_1,\ldots, \sigma^E_p$.
	
  Then, there exists $\varepsilon>0$ such that
  $[E-\varepsilon, E+\varepsilon]$ is a set of regular values of
  $\mu$, and $\mu^{-1}([E-\varepsilon, E+\varepsilon])$ is the
  disjoint union $\Sigma_1\sqcup\dots\sqcup\Sigma_p$ where each
  $\Sigma_q\subset\R$ is a compact interval containing $\sigma_q^E$ in
  its interior. Let $\varepsilon>0$ be any value satisfying the above
  conditions. For each $q=1,\dots, p$, let $\tilde \Sigma_q$ be an
  open interval containing $\Sigma_q$ such that the $\tilde\Sigma_q$'s
  are pairwise disjoint. Then the following holds.

  For each $q=1,\dots, p$, there exists a smooth map
  $\tilde\Sigma_q\ni \sigma\mapsto f_q(\sigma,\h)\in\R$ with an
  asymptotic expansion, in the smooth topology,
  \[
    f_q(\sigma,\h) \sim f_{q,0}(\sigma) + \h f_{q,1}(\sigma) + \h^2
    f_{q, 2}(\sigma) + \cdots
  \]
  depending only on the symbol of $\mathsf{P}_\h$ in the cylinder
  $\T_{2L}\times \Sigma_q$, such that the spectrum of $\mathsf{P}_\h$
  inside $[E-\varepsilon, E+\varepsilon]$ coincides, modulo
  $\O(\h^\infty)$, with the disjoint union
  \[
    \left(\bigsqcup_{q=1}^p \left\{f_q(\sigma,\h), \, \sigma \in \h
        (\tfrac{\pi}{L}\Z+\theta (\h^2)) \cap
        \tilde\Sigma_q\right\}\right) \cap [E-\varepsilon,
    E+\varepsilon]\,,
  \]
  see Definition~\ref{defi:coincide}. Moreover, we have
  \begin{align}
    f_{q,0}(\sigma)
    & = \mu(\sigma)_{| \Sigma_q} \label{equ:action}\\
    f_{q,1}(\sigma)
    & = \frac{-C(\sigma)_{| \Sigma_q}}{2L}\int_0^{2L}\kappa(s)\textup{d}s\,.
      \label{equ:kappa}
  \end{align}
\end{proposition}

Combining Proposition \ref{prop.pseudo1D} and Theorem \ref{thm.main},
we get the following result, where we use the notation of
Corollary~\ref{coro:dispersion} and Theorem~\ref{thm.main}.
\begin{corollary}[Spectrum of $\mathscr{L}_h$ at regular
  values]\label{cor.reg}
  Let $[a,b]$ be an interval disjoint from $\Theta$ and $\Lambda$.
  For each $k=1,\dots N$, for each $q=1,\dots, p(k)$, let
  $\Sigma_{k,q}\subset \R$ be an interval such that
  $\mu_k(\gamma,\cdot)$ is a diffeomorphism from $\Sigma_{k,q}$ to a
  neighborhood of $[a,b]$, in such a way that all $\Sigma_{k,q}$ are
  pairwise disjoint and $\bigcup_{q=1}^{p(k)}\Sigma_{k,q}$ contains
  $\mu_k(\gamma,\cdot)^{-1}([a,b])$. Then there exists a smooth map
  $\Sigma_{k,q}\ni \sigma\mapsto f_{k,q}(\sigma,\h)\in\R$ with an
  asymptotic expansion (in the smooth topology)
  \[
    f_{k,q}(\sigma,\h) \sim f_{k,q,0}(\sigma) + \h f_{k,q,1}(\sigma) +
    \h^2 f_{k,q, 2}(\sigma) + \cdots
  \]
  such that the spectrum of $\mathscr{L}_h$ in $[h a, h b]$ coincides,
  modulo $\O(h^2)$, with the disjoint union
  \[
    \left(\bigsqcup_{k=1}^N \bigsqcup_{q=1}^{p(k)} \left\{ h f_{k,q}
        (\sigma,h^{\frac12}), \, \sigma \in
        h^{\frac12}(\tfrac{\pi}{L}\Z + \theta (h)) \cap \Sigma_{k,q}
      \right\} \right) \cap [ha, hb]\,.
  \]
  Moreover, we have, when $\sigma\in \Sigma_{k,q}$,
  \begin{align}
    f_{k,q,0}(\sigma)
    & = \mu_k(\gamma,\sigma) \label{equ:fkq0}\\
    f_{k,q,1}(\sigma)
    & = -\langle \kappa \rangle C_k(\sigma)
      \label{equ:fkq1}
  \end{align}
  where $\langle \kappa \rangle$ is the average curvature:
  \[
    \langle \kappa \rangle =
    \frac{1}{2L}\int_0^{2L}\kappa(s)\textup{d}s\,.
  \]
\end{corollary}
Since the leading terms~\eqref{equ:fkq0} and~\eqref{equ:fkq1} do not
depend on $q$ (apart from the domain of definition $\Sigma_{k,q}$) we
obtain that the spectrum of $\mathscr{L}_h$ in $[h a, h b]$ coincides,
modulo $\O(h^2)$, with the disjoint union
\begin{equation}
\bigsqcup_{k=1}^N \left\{ h \mu_k (\gamma,\sigma) -
  h^{\frac32}\langle \kappa \rangle \overset{\circ}{C_k}(\sigma),
  \quad \sigma \in h^{\frac12}(\tfrac{\pi}{L}\Z + \theta (h))
\right\} \cap [ha, hb]\,.\label{equ:regular-O(h2)}
\end{equation}
As a first application of this corollary, we obtain a very accurate
formula for the number of eigenvalues of $\mathscr{L}_h$ in $[ha,hb]$,
this number being much smaller than what the crude estimate~\eqref{eq.roughNLhIh} says:
\begin{theorem}[Precise Weyl formula]\label{theo:weyl}
  Let $I_h=[ha,hb]$ where $[a,b]$ is an interval disjoint from
  $\Theta$ and $\Lambda$.  Then the number of eigenvalues of
  $\mathscr{L}_h$ in $I_h$ is
  \[
    N(\mathscr{L}_h,I_h) = \left\lfloor  \frac{L}{\pi h^{1/2}}\sum_{k,q}
    \delta_{k,q}^{[0]} + \frac{L\langle \kappa \rangle}{\pi}\sum_{k,q}
    \delta_{k,q}^{[1]} + \O(h^{1/2})\right \rfloor\,,
  \]
  where we use the notation $\sum_{k,q}:=\sum_{k=1}^N\sum_{q=1}^{p(k)}$, and
  \[
    \delta_{k,q}^{[0]} := \abs{\alpha_{k,q} - \beta_{k,q}}\,, \qquad
    \delta_{k,q}^{[1]} :=
    \frac{C_k(\beta_{k,q})}{\abs{\mu_k'(\beta_{k,q})}} -
    \frac{C_k(\alpha_{k,q})}{\abs{\mu_k'(\alpha_{k,q})}}\,,
  \]
  with $\alpha_{k,q}:=\mu_{k,q}^{-1}(a)$,
  $\beta_{k,q}:=\mu_{k,q}^{-1}(b)$.
\end{theorem}

In this statement we have denoted
$\mu_{k,q} := \mu_k(\gamma,\cdot)_{| \Sigma_{k,q}}$.  Notice that,
since the remainder term $\O(h^{1/2})$ tends to 0, we obtain that,
when $h$ is small enough, $N(\mathscr{L}_h,I_h)$ is equal to the
integer part of
$\frac{L}{\pi h^{1/2}}\sum_{k,q} \delta_{k,q}^{[0]} + \frac{L\langle
  \kappa \rangle}{\pi}\sum_{k,q} \delta_{k,q}^{[1]}$, or this plus or
minus 1.

\medskip

In a second application, we focus on the regular eigenvalues of
$\mathscr{L}_h$ below the first Landau level, and investigate how the
eigenvalues move when $h$ varies (by the scaling mentioned in the
introduction, this corresponds to the variation of the quantum
energies when the external magnetic field is modified). This variation
of eigenvalues is mainly due to the strong flux term
$\theta(h) = \frac{|\Omega|}{|\partial\Omega|h}$,
see~\eqref{equ:regular-O(h2)}. When $\gamma\in\R$, the eigenvalues
below the first Landau level are described by only two intervals
$\Sigma_{1,1}$ and $\Sigma_{1,2}$, for which the sense of variation of
the approximate eigenvalues with respect to $h$ are opposite. Hence,
we obtain a strongly oscillating behavior for these eigenvalues, which
is a generalization to excited states of the Little-Parks effect,
see~\cite{fournais2015lack}.

\begin{theorem}[Magnetic quantum oscillations]\label{theo:oscillate}
  Let $\gamma\in\R$.  Let $I_h=[ha, hb]$ with
  $a> \Theta^{{0}}(\gamma)$ and $b<1$. There exists $h_0>0$, $C>0$ and
  $M>0$ such that the following holds. Let $h<h_0$, and let $j\in\N$
  be such that the $j$-th eigenvalue $\lambda_j(\gamma,h)$ of
  $\mathscr{L}_h$ belongs to $I_h$. Then there exists $C_i=C_i(j,h)$,
  $i=1,2,3$, with $0<C_1<C_2<C_3\leq M$ such that, letting
  $h_i:= h+ C_ih^{2}$, we have
  \begin{itemize}
  \item $\lambda_j(h_2) \geq \lambda_j(h_1) + C h^{3/2}$,
  \item $\lambda_j(h_2) \geq \lambda_j(h_3) + C h^{3/2}$\,.
  \end{itemize}
  Moreover, the gap between consecutive eigenvalues is --- (roughly)
  periodically with period $\O(h^2)$ --- smaller than any order in
  $h$, precisely: there exists $h'$ such that $\abs{h-h'}=\O(h^2)$ and
  $\lambda_j(h') - \lambda_{j+1}(h') = \O(h^2)$, and there exists
  $h''$ such that $\abs{h-h''}=\O(h^2)$ and
  $\lambda_{j+1}(h'') - \lambda_{j}(h'') \geq Ch^{3/2}$.
\end{theorem}

See also Figure~\ref{fig:branches}. The proof of this theorem is given
in Section~\ref{sec:proof-theor-oscillate}. We believe that this is
the first mathematical treatment of quantum magnetic oscillations for
excited states in the first Landau band.  In principle, similar
oscillations for eigenvalues between higher Landau levels could be
obtained in the same vein. However, the growing number of connected
components $\Sigma_{k,q}$ involved would make the analysis (and
statement) quite complicated.

\begin{remark}
  These applications illustrate the fact that Corollary \ref{cor.reg}
  gives a very accurate description of the spectrum of $\mathscr{L}_h$
  by providing us with explicit approximations of the eigenvalues in
  $[ha,hb]$ modulo $\mathscr{O}(h^2)$. When $\gamma=+\infty$
  (\emph{i.e.} in the Dirichlet case), it also improves the
  description given in \cite[Corollary 2.7]{GV21} concerned with a
  thin spectral window containing a regular value. Moreover, although
  our results are formulated in terms of approximation of the
  eigenvalues, the strategy, based on microlocal projections, leading
  to Theorem \ref{thm.main} can also be used to describe the
  eigenspaces of $\mathscr{L}_h$ in terms of those of
  $\mathfrak{M}_h$.
\end{remark}

\subsubsection{Critical values}
Our main theorem also applies to the case when the spectral window
contains a critical value, \emph{i.e.} an element of $\Theta$, see
Corollary~\ref{coro:dispersion} (such a critical value is the unique
non-degenerate global minimum of a unique dispersion curve, see
Proposition \ref{prop.dispersion}). To illustrate this, let us focus
on the low-lying eigenvalues. The following corollary improves
\cite[Theorem I.5, $\alpha=\frac12$]{K06} by establishing the spectral
asymptotics of the lowest eigenvalues and by exhibiting spectral gaps
of order $h^ {\frac74}$ instead of $h^{\frac32}$ in the case of
regular values for each given dispersion curve. It also extends to any
Robin parameter the result obtained by Fournais and Helffer in
\cite{FH06} when $\gamma=0$.

Once Theorem~\ref{thm.main} is applied and reduces the analysis to a
single $\h$-pseudo-differential operator, this corollary becomes
essentially an application of~\cite[Proposition 6.8]{D-VN21}, see
details in Section~\ref{sec:proof-coroll-minipuits}.

\begin{corollary}\label{cor.lowlying}
  Consider $\gamma\neq\gamma^{[0]}_0$ with $\gamma^{[0]}_0$ defined in
  Remark \ref{rem16}, and let
  $\epsilon=\textup{sign}(\gamma^{[0]}_0-\gamma) =
  \textup{sign}(C_1(\xi_0(\gamma)))$. Assume that $\epsilon\kappa$
  admits a unique maximum at $s_{\max}$, which is
  non-degenerate. Then, for all $j\geq 1$, uniformly when
  $j h^{\frac14}=o(1)$,
  \[
    \lambda_j(\gamma,h) = \Theta^{[0]}(\gamma)h -
    \kappa(s_{\max})C_1(\xi_0(\gamma))h^{\frac32} + \frac{
      h^{\frac74}(2j-1)}{2}\sqrt{k_2
      C_1(\xi_0(\gamma))\mu''_1(\gamma,\xi_0(\gamma))}
    +o(h^{\frac74})\,,
  \]
  with $k_2=-\kappa''(s_{\max})$, and where we recall that
  $\xi_0(\gamma)$ is given in Proposition \ref{prop.dispersion}.
\end{corollary}
\begin{remark} Let us end the description of our results with a few
  comments about consequences and extensions following from our
  approach.
  \begin{enumerate}[\rm (i)]
  \item Corollary \ref{cor.lowlying} describes the low-lying
    eigenvalues with some uniformity in $j$ (which was not the case in
    \cite{FH06}), in an interval of the form
    $(-\infty, \Theta^{[0]}(\gamma)h + Ch^{3/2}]$. On the other hand,
    Corollary \ref{cor.reg} gives the spectrum in any interval of the
    form $[ha,hb]$ with $a> \Theta^{{0}}(\gamma)$ and $b<1$. Hence we
    have a spectral interval between these two regimes which we don't
    describe here. But actually, by using refined spectral results for
    1D pseudodifferential operators, and in particular the strategy
    of~\cite{D-VN21} in the case where $\kappa$ is a Morse function,
    it should also be possible to close this gap. However, this would
    require an analysis of the hyperbolic singularities arising from
    the minima of $\epsilon\kappa$, where we expect both a
    concentration of the eigenfunctions and a higher density of
    eigenvalues.
  \item When $\gamma>\gamma^{[0]}_0$, Corollary \ref{cor.lowlying}
    shows that the eigenfunctions (associated with the low-lying
    eigenvalues) are concentrated near the points of minimal
    curvature. This contrasts with the Neumann case when the points of
    maximal curvature play the role of attractive wells. This
    phenomenon was not observed before, see Remark \ref{rem16}.

  \item The case $\gamma=\gamma^{[0]}_0$ is critical since
    $C_1(\xi_0(\gamma))=0$. However, our analysis can still be used by
    computing additional subprincipal terms in our effective operator
    method. A similar phenomenon has recently been observed in the
    study of the magnetic Dirac operator \cite[Section 8]{BLTRS23} and
    also in the analysis of the magnetic Schrödinger operator with
    discontinuous magnetic fields \cite{FHKR22}. In this case, we
    have, for all $j\geq 1$,
    \[
      \lambda_j(\gamma,h) = \Theta^{[0]}(\gamma)h
      +h^2\lambda_j(\mathscr{A}_h)+o(h^2)\,,
    \]
    where
    $\mathscr{A}_h=\frac{\partial^2_\sigma\mu(\gamma,\xi_0(\gamma))}{2}(
    D_s+\theta(h)-h^{-\frac12}\xi_0(\gamma))^2+C_\gamma\kappa^2(s)$,
    for some $C_\gamma\in\mathbf{R}$. In this transition regime, the
    effective operator is not semiclassical.
  \item When the curvature $\kappa$ is constant, in the case
    $\gamma\in\mathbf{R}$, we are in a degenerate situation rather
    similar to the case when $\gamma=\gamma^{[0]}_0$. Concerning the
    operators $m_k^W$ of Theorem~\ref{thm.main}, this case corresponds
    to~\cite[Proposition 6.4]{D-VN21}. We can prove an expansion in
    the form
    \[
      \lambda_j(\gamma,h) = \Theta^{[0]}(\gamma)h- \kappa
      C_1(\xi_0(\gamma))h^{\frac32}
      +h^2\lambda_j(\mathscr{A}_h)+o(h^2)\,.
    \]
    Here, the eigenvalues of $\mathscr{A}_h$ will generate magnetic
    oscillations, see~\cite[Theorem 2.2; $k=0$]{D-VN21}.  When
    $\gamma=0$ and $j=1$, a similar estimate is described in
    \cite[Theorem 5.3.1]{FH06}.
  \end{enumerate}
	
\end{remark}

\subsection{Organization of the article}
In Section \ref{sec.2}, we prove that the eigenfunctions associated
with eigenvalues of $\mathscr{L}_h$ in $[h a,h b]$ are exponentially
localized near the boundary of $\Omega$, see Proposition
\ref{prop.Agmon}. Note that the strategy used to derive this
localization deviates from the usual variational method (see, for
instance, \cite{Helffer88} or \cite[Prop. 4.7]{Raymond17}), which
fails since we want to consider eigenvalues between two consecutive
Landau levels. To overcome this issue, our strategy, which eventually
generalizes the variational method, is based on establishing the
bijectivity of the magnetic Laplacian between exponentially weighted
$L^2$ spaces. In Section \ref{sec.3}, by means of tubular coordinates
$(s,t)$ near the boundary and a rescaling $t=h^{\frac12}\tau$, we
introduce a model operator $\mathscr{N}_{\hbar}$ depending on the
effective semiclassical parameter $\hbar=h^{\frac12}$, acting on
$2L$-periodic functions and involving a flux term $\mathfrak{f}_0$,
see \eqref{eq.Nhbar} and \eqref{eq.pshbar}. We also show that the
eigenfunctions of $\mathscr{L}_h$ are roughly microlocalized in a
compact set of the phase space attached to the boundary, see
Proposition \ref{prop.agmonmicro}. This allows to prove that the
spectrum of $\mathscr{L}_h$ (between the Landau levels) is located
near that of $\mathscr{N}_\hbar$, see Proposition
\ref{prop.firstspectralred}. However, one will see that Proposition
\ref{prop.firstspectralred} is not directly useful to establish our
main theorems. It is rather a pretext to motivate the introduction of
$\mathscr{N}_\hbar$ and to describe the spectral estimates required to
prove that spectra coincide modulo $\mathscr{O}(h^\infty)$. Actually,
one will compare directly the spectrum of $\mathscr{L}_h$ to that of
an effective operator on the boundary of $\Omega$. For that purpose,
in Section \ref{sec.4}, we construct a Grushin problem in order to
invert the (pseudo)differential operator $\mathfrak{N}_\hbar$ (which
acts as $\mathscr{N}_\hbar$ with $\mathfrak{f}_0$ replaced by
$0$). This method, inspired by the works of Martinez and Sjöstrand
(and adapted to magnetic operators by Keraval, see \cite{Keraval}),
has recently shown its efficiency to describe the low-lying
eigenvalues of various magnetic operators (see, for instance,
\cite{BHR21} and \cite{FHKR22}). The novelty in the present paper is
to use it to tackle the description of larger eigenvalues for magnetic
Schrödinger operators with boundaries, when several dispersion curves
are involved (see Figure \ref{fig:dispersion}), and not only the first
one as in \cite{BHR21} or \cite{FHKR22}. In order to use this method,
we write a semiclassical expansion of $\mathscr{N}_\hbar$, see
Proposition \ref{prop.descriptionNhbar}. The principal operator symbol
is the de Gennes operator (with Robin condition), which can be
inverted in the spectral window $[a,b]$ up to considering an augmented
matrix involving the eigenfunctions of the de Gennes operator, see
Lemma \ref{lem.inv0}. This allows to build an approximate inverse of
an augmented version of $\mathfrak{N}_\hbar$ denoted by
$\mathrm{Op}^{\mathrm{W}}_\hbar\mathscr{P}_\hbar$, see
\eqref{eq.PWhbar} (and the left and right quasi inverses
\eqref{eq.invgauche} and \eqref{eq.invdroite}). Thanks to these quasi
inverses, the bijectivity of $\mathscr{L}_h-z$ is reduced to that of a
pseudodifferential operator on $\T_{2L}$ whose matrix symbol is
$M_\hbar$, modulo some remainders, see Proposition \ref{prop.LhtoMh}
where the eigenfunctions of $\mathscr{L}_h$ are directly used as
quasimodes for $M_\hbar^{\mathrm{W}}$. In Section \ref{sec.5}, we
perform the spectral analysis of $M_\hbar^{\mathrm{W}}$ by using that
the principal matrix symbol $M_0$ is diagonal with uniform gaps
between the diagonal entries. We deduce Proposition
\ref{prop.pseudo1D} and Corollary \ref{cor.lowlying}. In Appendix
\ref{sec.roughWeyl}, we recall the origin of the estimate
\eqref{eq.roughNLhIh}. Appendix \ref{sec.deGennes} is devoted to the
de Gennes operator with Robin conditions: a couple of known results
are recalled and useful new ones are established.

\section{Exponential localization near the boundary and consequences}\label{sec.2}
Let us consider a smooth function $\Phi_0 : \overline{\Omega}\to \R_+$
that coincides with $\mathrm{dist}(x,\partial\Omega)$ near
$\partial\Omega$, and which vanishes only on $\partial\Omega$. Such a
function can be constructed as follows. Let $\epsilon>0$ be such that
the $\epsilon$-neighborhood of $\partial\Omega$, which we call
$\Omega_1$, admits a trivialization by the geodesic exponential: in
other words $\Omega_1\simeq \T\times [-\epsilon,\epsilon]$ with
coordinates $(s,t)$, and for any $x(s,t)\in\Omega_1$, we have
$\mathrm{dist}(x,\partial\Omega)=\abs{t}$, and $t>0$ if
$x\in\Omega$. We denote by $t:\Omega_1\to\R$ the corresponding
(smooth) map $x\mapsto t$. Let $\Omega_0\subset \Omega$ be the
complementary set of the $\epsilon/2$-neighborhood of
$\partial\Omega$. Thus, $\Omega_0\cup\Omega_1$ is an open neighborhood
of $\overline\Omega$. Let $(\chi_0,\chi_1)$, be an associated partition
of unity. The function $\Phi_0:= \chi_0+t\chi_1$ meets our
requirements.

Next, we extend $\Phi_0$ to a smooth function on $\mathbf{R}^2$ that
also belongs to $W^{2,\infty}(\mathbf{R}^2)$.

The following proposition states that the eigenfunctions of
$\mathscr{L}_h$ associated with eigenvalues in $I_h$ are localized
near the boundary of $\Omega$. The estimates look like Agmon's
estimates, but they are not obtained via variational means as it is
the case in many magnetic settings. Here, they follow from resolvent
estimates using the distance to the Landau levels.
\begin{proposition}\label{prop.Agmon}
  There exist $\alpha>0,C>0,h_0>0$ such that for all $h\in(0,h_0)$ and
  all eigenfunctions $\psi$ associated with an eigenvalue in $I_h$, we
  have
  \begin{equation}\label{eq.Agmona}
    \int_{\Omega}e^{2\alpha\Phi_0(x)/h^{1/2}}|\psi(x)|^2\dd x \leq
    C\|\psi\|^2\,,
  \end{equation}
  and
  \begin{equation}\label{eq.Agmonb}
    \int_{\Omega}e^{2\alpha\Phi_0(x)/h^{1/2}}|(-ih\nabla-\mathbf{A})\psi|^2\dd x
    \leq Ch \|\psi\|^2\,,
  \end{equation}
\end{proposition}

\subsection{Preliminaries}
In the following, $\mathscr{L}^{\R^2}_h$ denotes the operator
$(-ih\nabla-\mathbf{A})^2$ acting on the Hilbert space
$L^{2}(\R^{2})$.  By using the gauge invariance, we assume in the
whole section that $\mathbf{A}=\frac{1}{2}(-x_2,x_1)$.  Due to our
choice of eigenvalue $\lambda$, we deduce that
$\mathscr{L}^{\R^2}_h-\lambda$ is bijective and that there exists
$C>0$ such that, for all $h>0$,
\[
  \|(\mathscr{L}^{\R^2}_h-\lambda)^{-1}\|\leq Ch^{-1}\,.
\]
More precisely, we can take $C = \min(|2n-3-a|,|2n-1-b|)^{-1}$.  We
let $\Phi=\alpha \Phi_{0} $, with $\alpha>0$ to be determined, and consider the conjugated operator
\begin{align}
  \mathscr{L}^{\Phi}_h:= e^{\Phi/h^{1/2}} \mathscr{L}^{\R^2}_h
  e^{-\Phi/h^{1/2}}
  & = (-ih\nabla-\mathbf{A}+ih^{\frac12}\nabla\Phi)^2 \\
  & = \mathscr{L}^{\R^2}_h +
    2ih^{\frac12}\nabla\Phi\cdot(-ih\nabla-\mathbf{A}) - h\abs{\nabla\Phi}^2
    - i h^{\frac32} \Delta\Phi\,.
  \label{equ:conjugated}
\end{align}
The following lemma tells us that the invertibility is preserved for
$\mathscr{L}^{\Phi}_h - \lambda$ if $\alpha$ is small enough.
\begin{lemma}\label{lem.bij}
  There exists $C>0$ such that for all $h>0$ and all $\alpha>0$,
  \begin{equation}\label{eq.H1resolvant}
    h^{\frac12}\|\nabla\Phi\cdot(-ih\nabla-\mathbf{A})(\mathscr{L}^{\R^2}_h-\lambda)^{-1}\|\leq C\alpha\,.
  \end{equation}
  In particular, $\mathscr{L}^{\Phi}_h-\lambda$ is bijective as soon
  as $\alpha\leq \alpha_0$ and $\alpha_0$ is chosen small enough. With
  such a choice of $\alpha_0$, there exists $C>0$ such that, for all
  $h>0$, and all $\alpha\leq \alpha_0$,
  \begin{equation}\label{eq.resopert}
    \left\|\left(\mathscr{L}^{\Phi}_h-\lambda\right)^{-1}\right\|
    \leq \frac{C}{h}\,.
  \end{equation}
\end{lemma}

\begin{proof}
  Consider $v\in L^2(\R^2)$ and let
  $u=(\mathscr{L}^{\R^2}_h-\lambda)^{-1}v$. We have
  \[(\mathscr{L}^{\R^2}_h-\lambda)u=v\,,\] so that, by taking the
  scalar product with $u$ and using that $\lambda\leq Ch$,
  \[
    \|(-ih\nabla-\mathbf{A})u\|^2\leq Ch\|u\|^2+\|u\|\|v\|\,.
  \]
  Therefore, since $\nabla\Phi_0\in L^\infty$, there is a new constant
  $C'>0$ such that
  \[
    h^{\frac12}\|\nabla\Phi\cdot(-ih\nabla-\mathbf{A})u\|\leq C'\alpha
    h\|u\|+C'h^{\frac12}\alpha\|u\|^{\frac12}\|v\|^{\frac12}\,.
  \]
  Since
  $\|u\|=\left\|(\mathscr{L}^{\R^2}_h-\lambda)^{-1}v\right\|\leq
  Ch^{-1}\|v\|$, we see that
  \[
    h^{\frac12}\|\nabla\Phi\cdot(-ih\nabla-\mathbf{A})(\mathscr{L}^{\R^2}_h-\lambda)^{-1}v\|\leq
    \tilde C\alpha\|v\|\,,
  \]
  which gives \eqref{eq.H1resolvant}.
  
  Let us now deal with the bijectivity. We have
  \[
\mathscr{L}^{\Phi}_h - \lambda = \mathscr{L}^{\R^2}_h -
    \lambda + B
\]
with
\[
B :=  2ih^{\frac12}\nabla\Phi\cdot(-ih\nabla-\mathbf{A}) - h\abs{\nabla\Phi}^2
    - i h^{\frac32} \Delta\Phi\,.
\]
Since $\nabla\Phi_0$ and $\Delta\Phi_0$ are bounded, we deduce from \eqref{eq.H1resolvant} that
\[
  \|B(\mathscr{L}^{\R^2}_h-\lambda)^{-1}\|\leq
  C\alpha + C_1\alpha + h^{1/2}C_2\alpha^2 \leq \tilde C\alpha\,,
\]
when $\alpha$ is small enough. On the other hand,
\[
\mathscr{L}^{\Phi}_h - \lambda = \left(\mathrm{Id} + B(\mathscr{L}^{\R^2}_h-\lambda)^{-1}\right)(\mathscr{L}^{\R^2}_h-\lambda)\,;
\]
For $\alpha$ small enough, we deduce that
$\mathrm{Id} + B(\mathscr{L}^{\R^2}_h-\lambda)^{-1} $ is invertible,
and thus so is $\mathscr{L}^{\Phi}_h - \lambda $.
\end{proof}

In order to prove Proposition \ref{prop.Agmon}, we need to localize on
an $h^{1/2}$-neighborhood of $\partial\Omega$. For this purpose, we
introduce two functions $\chi_h\in\mathscr{C}^\infty_0(\Omega)$ and
$\widetilde \chi_h\in\mathscr{C}^\infty(\overline{\Omega})$ as
follows.
\[
  \chi_h \colon
  \begin{cases}
    \Omega \longrightarrow [0,1]\\
    x \longmapsto g(\Phi_0(x)/h^\frac{1}{2})
  \end{cases}
  \text{ and }  \widetilde \chi_h \colon
  \begin{cases}
     \overline\Omega \longrightarrow [0,1]\\
     x \longmapsto 1 - g(\Phi_0(x)/2h^\frac{1}{2})
  \end{cases}
\]
where $g$ is a smooth non-decreasing function on $\mathbf{R}$, valued
in $[0,1]$, equal to $0$ on $(-\infty,1)$ and to $1$ on
$(2,+\infty)$. In particular,
\begin{equation}\label{supp.cutoff1}
  {\rm supp}(\chi_h)
  \cap \overline{\Omega}\subset\{x\in\Omega\,,h^{-\frac12}\Phi_0(x)\geq 1\}\,,
\end{equation}
and
\begin{equation}\label{supp.cutoffs}\begin{split}
    {\rm supp}(\nabla \chi_h)\cap \overline{\Omega}
    &\subset \{x\in \Omega\,: h^{-\frac{1}{2}}\Phi_0(x)\in[1,2]\}
    \subset \{x\in \Omega\,:\widetilde \chi_h(x) = 1\}\,.
  \end{split}
\end{equation}

Note that the following properties hold:
\begin{enumerate}[---]
\item $\chi_h= 1$ away from an $h^{1/2}$-neighborhood
  $\partial\Omega$,
\item $\nabla\chi_h$ is supported in an $h^{1/2}$-neighborhood of
  $\partial\Omega$,
\item $\mathds{1}_{ {\rm supp}\nabla\chi_h}\leq \widetilde \chi_h$,
\item $\widetilde \chi_h= 0$ away from an $h^{1/2}$-neighborhood of
  $\partial\Omega$.
\end{enumerate}

\subsection{Proof of Proposition \ref{prop.Agmon}}
Let us consider $\lambda\in[h a,h b]\cap\mathrm{sp}(\mathscr{L}_h)$
and an associated eigenfunction $\psi\in
\mathrm{Dom}(\mathscr{L}_h)$. We have
\[\left((-ih\nabla-\mathbf{A})^2-\lambda\right)\psi=0\,.\]
Let $\varphi=e^{\Phi/h^{1/2}}\psi$. Using~\eqref{equ:conjugated} in
$\Omega$, the equation becomes
\begin{equation}\label{eq.eevarphi}
  \left(\mathscr{L}^{\Phi}_h-\lambda\right)\varphi=0\,.
\end{equation}
Then, we have
\begin{align}
  \left(\mathscr{L}^{\Phi}_h-\lambda\right)(\chi_h\varphi)
  & =[\mathscr{L}^{\Phi}_h,\chi_h]\varphi \nonumber\\
  & = e^{\Phi/h^{1/2}} [\mathscr{L}_h,\chi_h] e^{-\Phi/h^{1/2}}\varphi \nonumber\\
  & = e^{\Phi/h^{1/2}} \left(-h^2\Delta\chi_h -
    2ih\nabla\chi_h\cdot(-ih\nabla-\mathbf{A})
    \right)e^{-\Phi/h^{1/2}}\varphi\nonumber\\
  & = \left(-h^2\Delta\chi_h -
    2ih\nabla\chi_h\cdot(-ih\nabla-\mathbf{A}+ih^{\frac12}\nabla\Phi)
    \right)\varphi
    \label{equ:lphih-on-chi}
\end{align}
We have $\|h^2(\Delta\chi_h)\varphi\|\leq
Ch\|\tilde\chi_h\varphi\|$. (Here and in the rest of the paper, $C$
denotes a constant that is independent on $h$ but that can vary from
line to line.)  Let us explain how to deal with the last term. We have
\[
  h\|\nabla\chi_h\cdot(-ih\nabla-\mathbf{A}+ih^{\frac12}\nabla\Phi)\varphi\|
  \leq h\|\nabla\chi_h\cdot(-ih\nabla-\mathbf{A})\varphi\| +
  Ch\|\tilde\chi_h\varphi\| \,.
\]
Let us temporarily admit that, for
$\alpha$ small enough,
\begin{equation}\label{eq.nablachihnabla}
  h\|\nabla\chi_h\cdot(-ih\nabla-\mathbf{A})\varphi\|
  \leq Ch\|\tilde\chi_h\varphi\|\,.
\end{equation}
We then immediately deduce from~\eqref{equ:lphih-on-chi} that
\begin{equation}
  \label{302}
  \|\left(\mathscr{L}^{\Phi}_h-\lambda\right)(\chi_h\varphi)\|
  \leq \tilde Ch\|\tilde\chi_h\varphi\|\,.
\end{equation}
Since $\chi_h\varphi\in\mathrm{Dom}(\mathscr{L}^{\R^2}_h)$ we
obtain from~\eqref{eq.resopert} that
\[
  \|\chi_h\varphi\|\leq C\|\tilde\chi_h\varphi\|\,,
\]
which implies that
\[
  \|\varphi\|\leq  \tilde C(\|\tilde\chi_h\varphi\|+\|(1-\chi_h)\varphi\|)\,,
\]
showing that $\varphi$ is localized near $\partial\Omega$. More
precisely, recalling that $\varphi=e^{\Phi/h^{\frac12}}\psi$, using
that $\Phi_0(x)=\mathrm{dist}(x,\partial\Omega)$ near the boundary,
and the fact that the supports of $\tilde\chi_h$ and $1-\chi_h$ lie in
neighborhood of the boundary of size $h^{\frac12}$, we deduce
\eqref{eq.Agmona}.

 Let us now deal with \eqref{eq.Agmonb}. We have the Agmon identity
 \[
   \Re\,\langle\mathscr{L}_h\psi,\,e^{2\Phi/h^{1/2}}\psi\rangle
   = \mathscr{Q}_{h,\mathbf{A}}(e^{\Phi/h^{1/2}}\psi)
   - h\|e^{\Phi/h^{1/2}}\psi\nabla\Phi\|^2\,,
 \]
 which follows from~\eqref{equ:conjugated} where we see that
 $\Re\,\mathscr{L}^{\Phi}_h = \mathscr{L}_h - h\abs{\nabla\Phi}^2 $
 and we notice that
 $ \Re\,\langle\mathscr{L}_h\psi,\, e^{2\Phi/h^{1/2}}\psi\rangle =
 \langle (\Re\,\mathscr{L}^{\Phi}_h) e^{\Phi/h^{1/2}}\psi, \,
 e^{\Phi/h^{1/2}}\psi\rangle$. Recall also that, when
 $u\in\textup{Dom}(\mathscr{L}_h)$, then
 $\langle \mathscr{L}_h u, \, u \rangle =
 \mathscr{Q}_{h,\mathbf{A}}(u) $, see~\eqref{eq:Q}.

 Then, by using that $\psi$ is an
 eigenfunction, we get
 \[
   \int_{\Omega}\left|(-ih\nabla-\mathbf{A})\left(
       e^{\Phi/h^{1/2}}\psi\right)\right|^2\mathrm{d}x + \gamma
   h^{\frac32}\int_{\partial\Omega}|e^{\Phi/h^{1/2}}\psi|^2\mathrm{d}s-h\|e^{\Phi/h^{1/2}}\psi\nabla\Phi\|^2=\lambda\|e^{\Phi/h^{1/2}}\psi\|^2\,.
 \]
 With \eqref{eq.Agmona}, we find
 \[\int_{\Omega}\left|(-ih\nabla-\mathbf{A})\left(e^{\Phi/h^{1/2}}\psi\right)\right|^2\mathrm{d}x
   + \gamma
   h^{\frac32}\int_{\partial\Omega}|e^{\Phi/h^{1/2}}\psi|^2\mathrm{d}s\leq
   Ch\|\psi\|^2\,.\] From a classical trace theorem (see for
 instance~\cite[Section 5.5]{evans-book-pde}), there exists $C>0$ such
 that for all $\varepsilon>0$, we have
\[
  \int_{\partial\Omega}|\varphi|^2\mathrm{d}s\leq C\left(\varepsilon^{-1}\|\varphi\|^2+\varepsilon\|\nabla|\varphi|\|^2\right)\,.
\]
With the diamagnetic inequality (see for instance \cite[Theorem
2.1.1]{FH10}), we deduce that
\[h^2\int_{\partial\Omega}|\varphi|^2\mathrm{d}s\leq
  C\left(h^2\varepsilon^{-1}\|\varphi\|^2+\varepsilon\|(-ih\nabla-\mathbf{A})\varphi\|^2\right)\,,\]
and then
\[h^{\frac32}\int_{\partial\Omega}|\varphi|^2\mathrm{d}s\leq C\left(h^{\frac32}\varepsilon^{-1}\|\varphi\|^2+\varepsilon h^{-\frac12}\|(-ih\nabla-\mathbf{A})\varphi\|^2\right)\,.\]
Taking $\varepsilon=\frac{h^{\frac12}}{2|c|C}$ implies that
\[\int_{\Omega}\left|(-ih\nabla-\mathbf{A})\left(e^{\Phi/h^{1/2}}\psi\right)\right|^2\mathrm{d}x\leq \tilde Ch\|\psi\|^2\,.\]
Computing a commutator gives \eqref{eq.Agmonb}.

It remains to explain why \eqref{eq.nablachihnabla} holds.
From~\eqref{equ:conjugated} we can write
\begin{equation}\label{eq.LPhi}
\mathscr{L}^\Phi_h=L_1+ih^{\frac12}L_2\,,
\end{equation}
with
\[L_1=(-ih\nabla-\mathbf{A})^2-h|\nabla\Phi|^2\,,\quad
  L_2=2\nabla\Phi\cdot(-ih\nabla-\mathbf{A}) - ih\Delta\Phi\,.\]
From \eqref{eq.LPhi} and \eqref{eq.eevarphi}, we get
\[(L_1-\lambda+ih^{\frac12}L_2)\varphi=0\,.\]
For $j=1,2$, we have 
\[\Re\langle (L_1-\lambda)\varphi,(\partial_j\chi_h)^2\varphi\rangle-h^{\frac12}\Im\langle L_2\varphi,(\partial_j\chi_h)^2\varphi\rangle=0\,.\]
Thanks to the classical localization formula (see, for instance, \cite[Prop. 4.2]{Raymond17}), we have
\begin{multline*}
\Re\langle (L_1-\lambda)\varphi,(\partial_j\chi_h)^2\varphi\rangle\\
=\|(-ih\nabla-\mathbf{A})[(\partial_j\chi_h)\varphi]\|^2-h\int_{\Omega}|\nabla\Phi|^2|(\partial_j\chi_h)\varphi|^2\dd x-\lambda\|\partial_j\chi_h \varphi\|^2-h^2\|\nabla(\partial_j\chi_h)\varphi\|^2\,.
\end{multline*}
Moreover,
\[\begin{split}
&|\Im\langle L_2\varphi,(\partial_j\chi_h)^2\varphi\rangle|\\
&=|\Im\langle (\partial_j\chi_h)L_2\varphi,(\partial_j\chi_h)\varphi\rangle|\\
&\leq|\Im\langle L_2((\partial_j\chi_h)\varphi),(\partial_j\chi_h)\varphi\rangle|+|\langle[L_2,\partial_j\chi_h]\varphi,(\partial_j\chi_h)\varphi\rangle|\\
&\leq Ch\|(\partial_j\chi_h)\varphi\|^2+C\alpha\|(-ih\nabla-\mathbf{A})(\partial_j\chi_h)\varphi\|\|\partial_j\chi_h\varphi\|+Ch\|\nabla(\partial_j\chi_h)\varphi\|\|\partial_j\chi_h\varphi\|\,.
\end{split}\]
Due to the properties of $\chi_h$, we have
\begin{multline*}
|\Im\langle L_2\varphi,(\partial_j\chi_h)^2\varphi\rangle|\leq C\|\tilde\chi_h\varphi\|+Ch^{-\frac12}\alpha\|(-ih\nabla-\mathbf{A})(\partial_j\chi_h)\varphi\|^2+Ch^{-\frac12}\alpha\|\tilde\chi_h\varphi\|^2\\
+Ch^{-\frac12}\|\tilde\chi_h\varphi\|^2\,.
\end{multline*}
Therefore,
\begin{equation*}
\|(-ih\nabla-\mathbf{A})[(\partial_j\chi_h)\varphi]\|^2\leq C\|\tilde\chi_h\varphi\|^2+C\alpha\|(-ih\nabla-\mathbf{A})(\partial_j\chi_h)\varphi\|^2\,.
\end{equation*}
Taking $\alpha$ small enough, we get
\[\|(-ih\nabla-\mathbf{A})[(\partial_j\chi_h)\varphi]\|^2\leq C\|\tilde\chi_h\varphi\|^2\,.\]
Computing a commutator, we get \eqref{eq.nablachihnabla}.

\section{An operator on a semi-cylinder}\label{sec.3}

\subsection{A model operator}
The exponential localization near the boundary at a scale of order
$h^{\frac12}$ given by Proposition \ref{prop.Agmon} invites us to use
the classical tubular coordinates $(s,t)$ near the boundary. We recall
that these coordinates are defined thanks to the map
\[
  \Gamma : \T_{2L}\times(0,t_0)\ni(s,t)\mapsto
  \Gamma(s)-t\mathbf{n}(s) \qquad \T_{2L} := \R/2L\Z\,,
\]
which is injective if $t_0$ is small enough. Its Jacobian is
$a(s,t)=1-t\kappa(s)$, where $\kappa$ is the curvature of the boundary
at the point $\Gamma(s)$. Here $\Gamma$ is a counterclockwise
parametrization by the curvilinear abscissa. Thus, $\Gamma$ induces a
smooth diffeomorphism between $\T_{2L}\times(0,t_0)$ and
$\Omega_{t_0}:=\Gamma(\T_{2L}\times(0,t_0))$.

By using \cite[Appendix F]{FH10}, we can check that the magnetic Laplacian acts locally near the boundary in these coordinates as
\begin{multline*}
\widetilde{\mathscr{L}}_h=a(s,t)^{-1}\left(-ih\partial_s-t+\mathfrak{f}_{0}+\kappa(s)\frac{t^2}{2}\right)a(s,t)^{-1}\left(-ih\partial_s-t+\mathfrak{f}_{0}+\kappa(s)\frac{t^2}{2}\right)\\
-h^2a(s,t)^{-1}\partial_ta(s,t)\partial_t\,,
\end{multline*}
in the ambient Hilbert space $L^2(a\dd s\dd t)$. Here $\mathfrak{f}_0=\frac{|\Omega|}{|\partial\Omega|}$. The boundary condition \eqref{eq.boundary} becomes
\[\partial_t\psi(s,0)=\gamma h^{-\frac12}\psi(s,0)\,.\]
Of course the operator $\widetilde{\mathscr{L}}_h$ is only defined near $t=0$. We would like to consider a global operator. This can be done  by inserting cutoff functions with respect to $t$. We let $\check t=t\zeta(h^{-\frac12+\eta}t)$ with $\eta\in(0,\frac12)$ and $\zeta$ a smooth cutoff function equal to $1$ near $0$.

Let us consider the differential operator acting as
\begin{multline*}
	\widetilde{\mathscr{L}}_h= a(s,\check t)^{-1}\left(-ih\partial_s-t+\mathfrak{f}_{0}+\kappa(s)\frac{\check t^2}{2}\right) a(s,\check t)^{-1}\left(-ih\partial_s-t+\mathfrak{f}_{0}+\kappa(s)\frac{\check t^2}{2}\right)\\
	-h^2 a(s,\check t)^{-1}\partial_t a(s,\check t)\partial_t\,,
\end{multline*}
on the domain
  \begin{align*}
    \mathrm{Dom}(\widetilde{\mathscr{L}}_h)&=\big\lbrace u\in L^{2}( \T_{2L}\times \R_{+}): -\partial_{t}^{2}u \in  L^{2}( \T_{2L}\times \R_{+}),\\&\quad \quad \left(-ih\partial_s-t+\mathfrak{f}_{0}\right)^{2}u \in L^{2}( \T_{2L}\times \R_{+}) , \partial_{t}u(\cdot,0)=\gamma h^{-\frac12} u(\cdot,0) \big\rbrace \,.   
  \end{align*}
The ambient Hilbert space is $L^2(a(s,\check t)\mathrm{d}s\mathrm{d}t)=L^2(\mathrm{d}s\mathrm{d}t)$, with $2L$-periodic condition with respect to $s$.

The exponential localization of the original eigenfunctions at the
scale $h^{\frac{1}{2}}$ near the boundary suggests to consider the
partial rescaling
\[(s,t)=(s,\hbar \tau)\,,\]
where $\hbar=h^{\frac12}$. We consider the new operator, acting in the ambient Hilbert space $L^2(\hat a_\hbar \mathrm{d}s\mathrm{d}\tau)=L^2(\mathrm{d}s\mathrm{d}\tau)$,
\begin{equation}\label{eq.hatLh}
 \widehat{\mathscr{L}}_h=\hat a_\hbar(s,\tau)^{-1}p_{s,\hbar}\hat a_\hbar(s,\tau)^{-1}p_{s,\hbar}
	-\hat a_\hbar(s,\tau)^{-1}\partial_\tau\hat a_\hbar(s,\tau)\partial_\tau\,,
\end{equation}
with
\begin{equation}\label{eq.pshbar}
	p_{s,\hbar}=-i\hbar\partial_s-\tau+\hbar^{-1}\mathfrak{f}_0+\hbar\kappa(s)\frac{\hat \tau^2}{2}\,,	
\end{equation}
and where $\hat a_\hbar(s,\tau)=1-\hbar\hat\tau\kappa$ with $\hat\tau=\zeta(\hbar^{2\eta}\tau)\tau$.

The boundary condition becomes
\[\partial_{\tau}\psi(s,0)=\gamma\psi(s,0)\,.\]
The domain is given by
\begin{align*}
	\mathrm{Dom}(\widehat{\mathscr{L}}_h)&=\big\lbrace u\in L^{2}( \T_{2L}\times \R_{+}): -\partial_{\tau}^{2}u \in  L^{2}( \T_{2L}\times \R_{+}),\\&\quad \quad \left(-i\hbar\partial_s-\tau+\hbar^{-1}\mathfrak{f}_{0}\right)^{2}u \in L^{2}(  \T_{2L}\times \R_{+}) , \partial_{\tau}u(\cdot,0)=\gamma u(\cdot,0) \big\rbrace \,.   
\end{align*}
In fact, it will even be more convenient to deal with the following operator
\begin{equation}\label{eq.Nhbar}
\mathscr{N}_\hbar=\hat a_\hbar(s,\tau)^{-1}p^{\Xi_0}_{s,\hbar}\hat a_\hbar(s,\tau)^{-1}p^{\Xi_0}_{s,\hbar}
	-\hat a_\hbar(s,\tau)^{-1}\partial_\tau\hat a_\hbar(s,\tau)\partial_\tau\,,
\end{equation}
where we recall that $\Xi_0$ was defined in~\eqref{equ:circle}, and
\begin{equation}\label{eq.pshbarXi}
  p^{\Xi_0}_{s,\hbar} :=
  \Xi_0(\cdot+\hbar^{-1}\mathfrak{f}_0)^{\mathrm{W}}
  - \tau + \hbar\kappa(s)\frac{\hat \tau^2}{2}\,.
\end{equation}
\begin{align*}
	\mathrm{Dom}(\mathscr{N}_\hbar)&=\big\lbrace u\in L^{2}(  \T_{2L}\times \R_{+}): -\partial_{\tau}^{2}u \in  L^{2}( \T_{2L}\times \R_{+}),\\&\quad \quad \tau^{2}u \in L^{2}(  \T_{2L}\times \R_{+}) , \partial_{\tau}u(\cdot,0)=\gamma u(\cdot,0) \big\rbrace \,.   
\end{align*}

\subsection{Microlocalization of the eigenfunctions of $\mathscr{L}_h$}\label{sec.K}

In fact, we can prove that the eigenfunctions of $\mathscr{L}_h$ associated with eigenvalues in $[h a,h b]$ are roughly microlocalized with respect to $\sigma+\hbar^{-1}\mathfrak{f}_0$, the (shifted) dual variable of $s$. In order to quantify this, we consider the compact set  
\begin{equation}\label{eq.U}
	K=\bigcup_{j\geq 1}\lbrace \sigma\in\R : \mu_{j}(\sigma)\in[a,b] \rbrace\subset[\sigma_{\min},\sigma_{\max}]=:\tilde K\,. 
\end{equation}
Note that $K$ is indeed compact due to the properties of the $\mu_j$ (tending to $+\infty$ in $-\infty$) and to the choice of $[a,b]$, which does not contain Landau levels (the limits of the $\mu_j$ in $+\infty$).

The following result establishes a rather rough microlocalization result (with respect to $\sigma$) for the eigenfunctions: it tells us that the eigenfunctions are microlocalized in the compact set $\tilde K$. To quantify this, we consider a smooth function $\Xi$ with values in $[0,1]$ such that $\Xi=0$ near $\tilde K$ and $1$ away from $\tilde K$.

We let $\hat\lambda=h^{-1}\lambda$.

\begin{proposition}\label{prop.agmonmicro}
Let us consider the eigenvalue equation $\mathscr{L}_h\psi=\lambda\psi$ for $\lambda\in[h a,h b]$. Then, 
\begin{equation}\label{eq.cutagmon}
\widehat{\mathscr{L}}_h\varphi=\hat\lambda\varphi+\mathscr{O}(h^{\infty})\|\psi\|\,.
	\end{equation}
with $\varphi=\hat\chi_\hbar\hat\psi$, where $\hat\psi=\psi\circ\Gamma(s,\hbar\tau)$ and $\hat\chi_\hbar(\tau)=\chi(\hbar^{\eta}\tau)$ for a smooth cutoff function $\chi$ equal to $0$ away from $\tau=0$.

Moreover, 		
\begin{equation}\label{eq.cutmicro}
\mathrm{Op}^{\mathrm{W}}_\hbar(\Xi(\sigma+\hbar^{-1}\mathfrak{f}_0)) \varphi=\mathscr{O}(\hbar^\infty)\|\psi\|\,.
\end{equation}
\end{proposition}
\begin{proof}
  The estimate \eqref{eq.cutagmon} follows from the localization near
  the boundary (Proposition~\ref{prop.Agmon}).

  Then, let us only prove that \eqref{eq.cutmicro} holds when $\Xi$ is
  $0$ near $(-\infty,\sigma_{\max}+\frac\epsilon2)$ and $1$ on
  $(\sigma_{\max}+\epsilon,+\infty)$, the estimate following from
  similar arguments on $(-\infty,\sigma_{\min}-\epsilon)$.

  In order to lighten the notation, we will use a slight abuse of
  notation by writing
  \begin{equation}\label{eq.XiW}
    \Xi^{\mathrm{W}}:=\mathrm{Op}^{\mathrm{W}}_\hbar(\Xi(\sigma+\hbar^{-1}\mathfrak{f}_0))\,.
  \end{equation}
  Then, we write
  \[\left(\widehat{\mathscr{L}}_h-\hat\lambda\right)\Xi^{\mathrm{W}}\varphi=[\widehat{\mathscr{L}}_h,\Xi^\mathrm{W}]\varphi+\mathscr{O}(\hbar^\infty)\|\psi\|\,.\]
  Thanks to the explicit expression \eqref{eq.hatLh}, we get
  \begin{equation}\label{eq.commuthatLXi}
    \|[\widehat{\mathscr{L}}_h,\Xi^{\mathrm{W}}]\varphi\|\leq C\hbar\|\underline{\Xi}^{\mathrm{W}} \varphi\|+C\hbar\|\underline{\Xi}^{\mathrm{W}}\partial_\tau\varphi\|+\mathscr{O}(\hbar^{\infty})\|\psi\|\,,
  \end{equation}
  and we can write, by using the support of
  $\chi(h^{-\frac12+\eta}t)$,
  \begin{equation}\label{eq.perturbationLh}
    \widehat{\mathscr{L}}_h = \widehat{\mathscr{L}}_0+\mathscr{R}_\hbar\,,
    \qquad\widehat{\mathscr{L}}_0=-\partial^2_\tau+p_{s,\hbar,0}^2\,,
    \qquad p_{s,\hbar,0}=-i\hbar\partial_s+\hbar^{-\frac12}\mathfrak{f}_0-\tau\,,
  \end{equation}
  where the remainder $\mathscr{R}_\hbar$ can be written as
  \begin{equation}\label{eq.Rhcut}
    \mathscr{R}_\hbar = \hbar^{1-2\eta}R_{\hbar,2}(s,\tau)p_{s,\hbar,0}^2 +
    \hbar^{1-4\eta}R_{\hbar,1}(s,\tau)p_{s,\hbar,0} +
    \hbar^{2-8\eta}R_{\hbar,3}+\hbar R_{\hbar,4}\partial_\tau\,,
  \end{equation}
  the $R_{\hbar,j}$ being smooth functions, uniformly bounded in
  $\hbar$.

  Then, we consider an increasing function
  $\sigma\mapsto\tilde\Xi(\sigma)\in(\sigma_{\max}+\frac\epsilon4,+\infty)$
  that coincides with $\mathrm{Id}$ on
  $(\sigma_{\max}+\frac\epsilon2,+\infty)$. We let
  \[
    \widehat{\mathscr{L}}^{\mathrm{cut}}_0=\mathrm{Op}^{\mathrm{W}}_\hbar\left(-\partial^2_\tau+(\tilde\Xi(\sigma+\hbar^{-1}\mathfrak{f}_0)-\tau)^2\right)\,,
  \]
  acting on $L^2(\T_{2L}\times \R_+)$, where the superscript
  "$\mathrm{cut}$" refers to the replacement of
  $-i\hbar\partial_s+\hbar^{-1}\mathfrak{f}_0$ by
  $\tilde\Xi^{\mathrm{W}}$ (with the same abuse of notation as in
  \eqref{eq.XiW}). We notice that
  $\widehat{\mathscr{L}}^{\mathrm{cut}}_0-\hat\lambda$ is bijective
  (with an inverse uniformly bounded in $\hbar$) due to the choice of
  $\tilde\Xi$ and the definition of $\sigma_{\max}$. Moreover, we have
  $\{\Xi\neq 0\}\subset\{\tilde\Xi=\mathrm{Id}\}$ so that, with
  \eqref{eq.perturbationLh},
  \[\left(\widehat{\mathscr{L}}^{\mathrm{cut}}_0-\hat\lambda+\mathscr{R}^{\mathrm{cut}}_\hbar\right)\Xi^{\mathrm{W}}\varphi=[\widehat{\mathscr{L}}_h,\Xi^\mathrm{W}]\varphi+\mathscr{O}(\hbar^\infty)\|\psi\|\,,\]
  which can be written as
  \[\left(\mathrm{Id}+\mathscr{R}^{\mathrm{cut}}_\hbar(\widehat{\mathscr{L}}^{\mathrm{cut}}_0-\hat\lambda)^{-1}\right)(\widehat{\mathscr{L}}^{\mathrm{cut}}_0-\hat\lambda)\Xi^{\mathrm{W}}\varphi=[\widehat{\mathscr{L}}_h,\Xi^\mathrm{W}]\varphi+\mathscr{O}(\hbar^\infty)\|\psi\|\,.\]
  By using \eqref{eq.Rhcut} and applying the Calder\'on-Vaillancourt
  theorem, we get that
  \[
    \|\mathscr{R}^{\mathrm{cut}}_\hbar(\widehat{\mathscr{L}}^{\mathrm{cut}}_0
    - \hat\lambda)^{-1}\| = \mathscr{O}(\hbar^{1-4\eta})\,.
  \]
  Thus, the operator
  $\mathrm{Id} +
  \mathscr{R}^{\mathrm{cut}}_\hbar(\widehat{\mathscr{L}}^{\mathrm{cut}}_0-\hat\lambda)^{-1}$
  is bijective as soon as $\hbar$ is small enough.

  With \eqref{eq.commuthatLXi}, this provides us first with
  \[\|\Xi^{\mathrm{W}}\varphi\|^2\leq C\hbar\|\underline{\Xi}^{\mathrm{W}} \varphi\|^2+C\hbar\|\underline{\Xi}^{\mathrm{W}}\partial_\tau\varphi\|^2\\
    +\mathscr{O}(\hbar^{\infty})\|\psi\|^2\,,\] and then
  \[\|{\Xi}^{\mathrm{W}}\partial_\tau\varphi\|^2+\|\Xi^{\mathrm{W}}\varphi\|^2\leq
    C\hbar\|\underline{\Xi}^{\mathrm{W}}
    \varphi\|^2+C\hbar\|\underline{\Xi}^{\mathrm{W}}\partial_\tau\varphi\|^2
    +\mathscr{O}(\hbar^{\infty})\|\psi\|^2\,.\] The estimate
  \eqref{eq.cutmicro} follows by induction on the size of the support
  of $\Xi$.
\end{proof}

\subsection{First spectral estimates}
The aim of the following proposition is to establish that the spectrum
of $\mathscr{L}_h$ in $I_h$ is close to that of $h\mathscr{N}_\hbar$
and thus that $\mathscr{N}_\hbar$ is a nice auxiliary operator to
describe the spectrum of $\mathscr{L}_h$. In fact, we will see that
this proposition is not necessary to prove our spectral estimates, but
its proof is instructive.

\begin{proposition}\label{prop.firstspectralred}
There exists $h_0>0$ such that for all $h\in(0,h_0)$ the following holds.
Let us consider an interval $J_h\subset I_h$. Then, there exists an interval $ \hat J_h$ such that $J_h\subset \hat J_h\subset I_h$ with $\mathrm{d}_{\mathrm{H}}(J_h,\hat J_h)=\mathscr{O}(h^\infty)$ and
\begin{equation}\label{eq.distmult}
\mathrm{rank}\,\mathds{1}_{J_h}(\mathscr{L}_h)\leq \mathrm{rank}\,\mathds{1}_{\hat J_h}(h\mathscr{N}_\hbar)\,.
\end{equation}
Moreover, for all $\lambda\in I_h\cap\mathrm{sp}(\mathscr{L}_h)$,
\begin{equation}\label{eq.distspec}
\mathrm{dist}(\lambda,h\mathrm{sp}({\mathscr{N}}_\hbar))=\mathscr{O}(h^\infty)\,.
\end{equation}
\end{proposition}

\begin{proof}
Let us start by proving \eqref{eq.distspec}. Let us consider an eigenvalue $\lambda\in I_h$ of $\mathscr{L}_h$.	We write the eigenvalue equation $\mathscr{L}_h\psi=\lambda\psi$. 

With Proposition \ref{prop.agmonmicro}, we can write \eqref{eq.cutagmon}. Then, with \eqref{eq.cutmicro}, we deduce that
\[h\mathscr{N}_\hbar\varphi=\lambda\varphi+\mathscr{O}(\hbar^{\infty})\|\psi\|\,.\]
Thus, \eqref{eq.distspec} follows from the spectral theorem.

Let us now consider \eqref{eq.distmult}, which deals with
multiplicities. Let us write
$\mathrm{sp}(\mathscr{L}_h)\cap J_h=\{\lambda_1,\ldots,\lambda_p\}$
(where the $\lambda_j$ are distinct) and underline that these
eigenvalues depend on $h$ as well as $p$. Consider the associated
eigenspaces $(E_j)_{1\leq j\leq p}$ and note that
$\dim\bigoplus_{j=1}^p E_j=\mathscr{O}(h^{-2})$ thanks to the Weyl
estimate \eqref{eq.roughNLhIh}. With the same notation as above, we
consider the spaces of quasimodes
$(\hat\chi_\hbar \hat E_j)_{1\leq j\leq p}$. Thanks to Proposition
\ref{prop.agmonmicro} (and the rough Weyl estimate),
$\dim (\hat\chi_\hbar \hat E_j)=\dim E_j$, as soon as $h$ is small
enough. Moreover, we have
\[\|(\oplus_{j=1}^ph\mathscr{N}_\hbar-\lambda)\varphi\|\leq\varepsilon_h\|\varphi\|\,,\quad \varepsilon_h=\mathscr{O}(h^\infty)\,,\]
for all $\varphi=(\varphi_1,\ldots,\varphi_p)\in \bigoplus_{j=1}^p \hat\chi_\hbar \hat E_j$ and where $\lambda=(\lambda_1,\ldots,\lambda_p)$. 

We set $J_h=[a_h,b_h]$ and $\hat J_h=[a_h-\varepsilon_h,b_h+\varepsilon_h]$. If $\mathrm{rank}\,\mathds{1}_{\hat J_h}(h\mathscr{N}_\hbar)<\mathrm{rank}\,\mathds{1}_{J_h}(\mathscr{L}_h)$, then the projection $\Pi : \bigoplus_{j=1}^p \hat\chi_\hbar \hat E_j\to \mathrm{ran}\,\mathds{1}_{\hat J_h}(h\mathscr{N}_\hbar)$ could not be injective. Considering a non-zero $\varphi$ in its kernel, the spectral theorem would give $\|(\oplus_{j=1}^p h\mathscr{N}_\hbar-\lambda)\varphi\|>\varepsilon_h\|\varphi\|$, which is a contradiction when $\varphi\neq 0$. Therefore, \eqref{eq.distmult} follows.
\end{proof}

\section{A Grushin problem}\label{sec.4}

\subsection{A pseudodifferential operator with operator-valued symbol}
Recalling Remark \ref{rem.period}, we notice that the operator $\mathscr{N}_\hbar$ can be seen as a pseudo-differential operator acting as
\begin{equation*}
\mathfrak{N}_\hbar=\hat a_\hbar(s,\tau)^{-1}	\mathcal T_{\hbar}\hat a_\hbar(s,\tau)^{-1}	\mathcal T_{\hbar}
	-\hat a_\hbar(s,\tau)^{-1}\partial_\tau\hat a_\hbar(s,\tau)\partial_\tau\,,
\end{equation*}
on functions of the form $e^{is\mathfrak{f}_0/h}L^2( \T_{2L}\times\R_+)$ and where
\begin{equation*}
	\mathcal T_{\hbar}= \Xi_0^\mathrm{W}-\tau+\hbar\frac{\kappa}{2}\hat\tau^2\,.
\end{equation*}

In fact, it will be convenient to see $\mathfrak{N}_\hbar$ as a
pseudo-differential operator with operator-valued symbol.  At a formal
level, the principal symbol of $\mathfrak{N}_\hbar$ is
$n_0(s,\sigma)=-\partial^2_\tau+(\Xi_0(\sigma)-\tau)^2$ equipped with
the domain
\[
  \mathrm{Dom}(n_0)=\{\psi\in B^2(\R_+) : \psi'(0)=c\psi(0)\}\,.
\]
The vector space $B^2(\R_+)$ is equipped with the  $(s,\sigma)$-independent norm
\[\|\psi\|^2_{B^2(\R_+)}=\|\psi''\|^2+\|\psi'\|^2+\|\langle t\rangle^2\psi\|^2\,.\]
With this convention, we may write that $n_0\in S(\R^2,\mathscr{L}(B^2(\R_+),L^2(\R_+)))$.

We say that $\Psi\in S(\R^2,\mathscr{L}(B^2(\R_+),L^2(\R_+)))$ when, for all $\alpha\in\N^2$, there exists $C_\alpha>0$ such that for all $(s,\sigma)\in\R^2$,
\[\|\partial^\alpha\Psi\|_{\mathscr{L}(B^2(\R_+),L^2(\R_+))}\leq C_\alpha\,.\]
Such symbols might also depend on $\hbar$; in this case, the constant $C_\alpha$ is uniform in $\hbar$.

\begin{lemma}\label{lem.Nhbarpseudo}
The operator $\mathfrak{N}_\hbar$ can be written as the Weyl quantization of a symbol  in $S(\R^2,\mathscr{L}(B^2(\R_+),L^2(\R_+)))$.
\end{lemma}
\begin{proof}
We can write
\[\mathcal{T}_\hbar=\mathrm{Op}^{\mathrm{W}}_\hbar\left(\Xi_0(\sigma)-\tau+\hbar\frac{\kappa}{2}\hat\tau^2\right)\,,\]
the symbol ($2L$-periodic with respect to $s$) belonging to the class  $S(\R^2,\mathscr{L}(B^1(\R_+),L^2(\R_+)))\cap S(\R^2,\mathscr{L}(B^2(\R_+),B^1(\R_+)))$. The functions $a_\hbar(s,\tau)$ and $a_\hbar(s,\tau)^{-1}$ are bounded uniformly with respect to $\hbar$ (and so are all their derivatives). Then, the conclusion follows from the composition theorem for pseudo-differential operators, see \cite[Theorem 2.1.12]{Keraval}.
\end{proof}

In the following, we let $\mu=\hbar^{2\eta}$ and
$\zeta_\mu(\tau)=\zeta(\mu\tau)$. This is convenient when expanding
the operator in powers of $\hbar$ ($\mu$ will be considered a
parameter). This expansion allows to describe rather accurately the
symbol of $\mathfrak{N}_\hbar$ by expanding it in powers of
$\hbar$. An analogous description for a very similar operator can be
found in great detail in \cite[Section 4.2]{FHKR22}.

\begin{proposition}\label{prop.descriptionNhbar}
The operator $\mathfrak{N}_\hbar$ can be written as follows:
\begin{equation}\label{eq:ex-Nhc}
	\mathfrak{N}_{\hbar}=\mathfrak n_0+\hbar\mathfrak n_1+\mathcal \hbar^2\mathscr R_\hbar^{(2)}+\hbar w_{\hbar}\partial_\tau\,,
\end{equation}
where, for some $N\in\mathbb {N}$, $C, \hbar_0>0$, we have, for all $\hbar\in(0,\hbar_0)$,

\begin{enumerate}[\rm (i)]
	\item\label{eq.i} $w_{\hbar}$ is a smooth function supported in $\{(s,\tau) : C^{-1}\hbar^{-2\eta}\leq \langle \tau\rangle\leq C\hbar^{-2\eta}\}$ and such that $w_{\hbar}=\mathscr{O}(\langle \tau\rangle)$ ,
	\item \label{eq.ii}$\mathscr R_\hbar^{(2)}$  is a pseudodifferential operator whose symbol belongs to a bounded set in the space of symbols $S(\R^2,\mathscr{L}(B^2(\R_+),L^2(\R_+,\langle \tau\rangle^{-N}\mathrm{d}\tau)))$.
\end{enumerate}
Moreover, the $\mathfrak{n}_j$ are given by $\mathfrak{n}_j=\mathrm{Op}^{\mathrm{W}}_\hbar n_j$ with
\begin{equation}
	\begin{split}
		n_0&=-\partial^2_\tau+(\Xi_0(\sigma)-\tau)^2\,,\\
		n_1&=\kappa(s)\left[(\Xi_0(\sigma)-\tau) \zeta_\mu^2\tau^2+\zeta_\mu\partial_\tau+2\zeta_{\mu}\tau(\Xi_0(\sigma)-\tau)^2\right]\,.
	\end{split}
\end{equation}

In particular, we can write $\mathfrak{N}_\hbar=\mathrm{Op}^{\mathrm{W}}_\hbar(n_\hbar)$ with a symbol $n_\hbar$ satisfying
\[n_\hbar=n_0+\hbar n_1+\hbar^2 r_\hbar^{(2)}+\hbar w_{\hbar}\partial_\tau\,,\]
where  $r_\hbar^{(2)}$ belongs to the class of operator symbols $S(\R^2,\mathscr{L}(B^2(\R_+),L^2(\R_+,\langle \tau\rangle^{-N}\mathrm{d}\tau)))$ uniformly in $\hbar$.
\end{proposition}

\subsection{Dimensional reduction}
The aim of this section is to analyse the spectrum of ${\mathscr{N}}_\hbar$. This can be done thanks to a Grushin reduction. The principal symbol of ${\mathfrak{N}}_\hbar$ is the "de Gennes operator" with Robin  boundary conditions. Explicitly, we have
\[n_0(s,\sigma)=-\partial^2_\tau+(\Xi_0(\sigma)-\tau)^2\,.\] The
increasing sequence of its (simple) eigenvalues is
$(\mu_k(\Xi_0(\sigma)))_{k\geq 1}$. We recall that the functions
$\mu_k$ are described in Proposition \ref{prop.dispersion}.

Now, consider the window $[a,b]\subset (2n-3,2n-1)$. For simplicity,
let us denote
$\overset{\circ}{u}_k := \overset{\circ}{u}_k^{[\gamma,\sigma]}$, see
Section~\ref{sec:de-gennes-operator} and \eqref{equ:circle}. Let $N$
be defined as in~\eqref{equ:N}.

\begin{lemma}\label{lem.inv0}
For all $z\in[a,b]$, let us consider the matrix operator
\[\mathscr{P}_0(z)=\begin{pmatrix}
n_0(s,\sigma)-z&\Pi^*\\
\Pi&0
\end{pmatrix} : B^2(\R_+)\times \C^N\longrightarrow L^2(\R_+)\times\C^N\,,\]
where $\Pi^*(\alpha)=\sum_{j=1}^N\alpha_j \overset{\circ}{u}_j$ and $\Pi\psi=(\langle\psi,\overset{\circ}{u}_j\rangle)_{1\leq k\leq N}$.

Then, $\mathscr{P}_0(z)$ is bijective with inverse 
	\[\mathscr{Q}_0(z)=\begin{pmatrix}
	q_0&\Pi^*\\
	\Pi&z-M_0(\sigma)
	\end{pmatrix}\,,\qquad q_0=(n_0(s,\sigma)-z)^{-1}(\Pi^*\Pi)^{\perp}\,,\]
	where $M_0(\sigma)$ is the diagonal $N\times N$ matrix whose diagonal is $(\overset{\circ}{\mu}_1,\ldots,\overset{\circ}{\mu}_N)$.
\end{lemma}
\begin{proof}
Let $g \in  L^{2}(\R^+)$ and $\beta \in \C^N\,.$ Let us look for $f\in \mathrm{Dom}(n_{0})$ and $\alpha\in \C^N $ such that $$\mathscr{P}_0(z)( f \oplus \alpha)=g \oplus \beta\,. $$
 In other words, 
 \[ (n_0(s,\sigma)-z)f + \Pi^* \alpha=g ,\quad  \Pi f=\beta\,. \]
 Let $E= \mathrm{span}(\overset{\circ}{u}_1, \ldots, \overset{\circ}{u}_N),  $ and $F=E^\perp\,.$
 We can write $f=f_E+f_F$
 where 
 \[f_E= \sum _{j=1}^{N} \langle f, \overset{\circ}{u}_j \rangle \overset{\circ}{u}_j=\Pi^*\Pi f\,,\quad  f_F=   (\Pi ^* \Pi )^{\perp}f  \,.\]
We have 
 \begin{align*}
      (n_0(s,\sigma)-z)f_F&= - (n_0(s,\sigma)-z)f_E   - \Pi^* \alpha+g \\&= - (n_0(s,\sigma)-z)\sum _{j=1}^{N} \beta_j \overset{\circ}{u}_j   - \Pi^* \alpha+g \,,
 \end{align*} 
so that
 \begin{equation}
 \label{500}
    (n_0(s,\sigma)-z)f_F =  - \sum _{j=1}^{N} \beta_j (\overset{\circ}{\mu}_j-z)  \overset{\circ}{u}_j   - \Pi ^*\alpha+g  \,. 
\end{equation}
 The space $F$  is stable by $ n_0(s,\sigma)-z$.
 
Moreover, thanks to the self-adjointness of $n_0$, the min-max principle and the fact that $\min \overset{\circ}{\mu}_{N+1} \geq 2N+1 > z$,  there exists $c>0$ such that, for all $u\in \mathrm{Dom}(n_0)\cap F$,
\[\langle (n_0-z) u, u\rangle=  \langle n_0 u, u \rangle -z\|u\|^2 \geq  (\overset{\circ}{\mu}_{N+1}-z)\|u\|^2 \geq  c \|u\|^2\,.\]
Thus, the operator $(n_0-z)_{|F} $ is injective with closed range and,
by considering the adjoint, we deduce that it is bijective. We also
notice that
 \[\| (n_0-z)^{-1}_{|F}\|  \leq ( \overset{\circ}{\mu}_{N+1}-z)^{-1}\leq c^{-1}\,.\]
 Then \eqref{500} has a solution if and only if the right-hand-side belongs to $F$, that is
 \begin{equation*}
  - \sum _{j=1}^{N} \beta_j (\overset{\circ}{\mu}_j-z)  \overset{\circ}{u}_j   - \Pi^* \alpha+g \in F 
 \end{equation*}
 which means that, for all $k\in\{1,\ldots,N\}$,
 \[-  \beta_k (\overset{\circ}{\mu}_k-z)    -  \alpha_{k} + \langle g  , \overset{\circ}{u}_k \rangle =0\,.\]
 We deduce that
 \[ \alpha= \Pi g+( z-M_0(\sigma) )\beta.\, 
 \]
 This unique solution is given by 
 \begin{equation*}
  f_F =(\Pi^*  \Pi )^{\perp} (n_0(s,\sigma)-z)^{-1} g    \,.     
 \end{equation*}
 Therefore, 
 \[f=   f_E+ f_F=\Pi \beta + (\Pi^*  \Pi )^{\perp} (n_0(s,\sigma)-z)^{-1} g \,. \]
\end{proof}

Let us now consider the full symbol
\[\mathscr{P}_\hbar(z):=\begin{pmatrix}
{n}_\hbar-z&\Pi^*\\
\Pi&0
\end{pmatrix}\,,\]
which may be expanded in powers of $\hbar$ as
\[\mathscr{P}_\hbar(z)=\mathscr{P}_0(z)+\hbar\mathscr{P}_1(z)+ \mathscr{R}_\hbar\,,\]
with 
\[\mathscr{P}_0(z)=\begin{pmatrix}
n_0(s,\sigma)-z&\Pi^*\\
\Pi&0
\end{pmatrix}  \,,\quad \mathscr{P}_1(z)=\begin{pmatrix}
n_{1}&0\\
0&0
\end{pmatrix} \,, \quad \mathscr{R}_\hbar(z)=\begin{pmatrix}
r_{\hbar}&0\\
0&0
\end{pmatrix}\,,
\]
where $r_\hbar=\hbar^2 r_\hbar^{(2)}+\hbar \tilde w_{\hbar,1}+\hbar w_{\hbar,2}\partial_\tau$, see Proposition \ref{prop.descriptionNhbar}.

We notice that
\begin{equation}\label{eq.PWhbar}
\mathscr{P}^{\mathrm{W}}_\hbar=\begin{pmatrix}
\mathfrak{N}_\hbar-z&\mathfrak{P}^*\\
\mathfrak{P}&0
\end{pmatrix}\,,\qquad \mathfrak{P}=\Pi^{\mathrm{W}}\,.
\end {equation}
Since the principal symbol of $\mathscr{P}^{\mathrm{W}}_\hbar$ is
bijective, it is natural to try to construct an approximate inverse in
the semiclassical limit. Let us look for an approximate inverse whose
symbol is in the form
\begin{equation}\label{eq.Qhbar,1}
\mathscr{Q}_{\hbar,1}=\mathscr{Q}_0(z)+\hbar \mathscr{Q}_1(z)\,.
\end{equation}
As in \cite{Keraval}, we are led to choose
\[\mathscr{Q}_1=-\mathscr{Q}_0\mathscr{P}_1\mathscr{Q}_0=-\begin{pmatrix}
	q_0n_1 q_0&q_0 n_1 \Pi^*\\
	\Pi n_1 q_0&\Pi n_1 \Pi^*
\end{pmatrix}\,.\]
This choice is convenient since the composition theorem for pseudo-differential operators (see \cite{Keraval}) implies that
\begin{multline*}
\mathscr{Q}^{\mathrm{W}}_{\hbar,1}\left(\mathscr{P}_0(z)+\hbar\mathscr{P}_1\right)^{\mathrm{W}}=\mathrm{Id}+\hbar\left(\frac{1}{i}\{\mathscr{Q}_0,\mathscr{P}_0\}+\mathscr{Q}_0\mathscr{P}_1+\mathscr{Q}_1\mathscr{P}_0\right)^{\mathrm{W}}\\
+\mathscr{O}_{L^2( \T_{2L}\times\R_+,\langle\tau\rangle^N\mathrm{d}s\mathrm{d}\tau)\times L^2( \T_{2L})\to L^2( \T_{2L}\times\R_+)\times L^2( \T_{2L})}(\hbar^2)\,,
\end{multline*}
where the remainder is estimated thanks to the Calder\'on-Vaillancourt
theorem (see \cite[Theorem 2.1.16]{Keraval}) and the resolvent
estimate in Lemma \ref{lem.sandwich} (applied with an appropriate
$\alpha>0$). The $\hbar$-term vanishes due the choice of
$\mathscr{Q}_1$ and that fact that the Poisson bracket is actually $0$
since the principal symbol does not depend on $s$.  With this choice,
the bottom right coefficient, denoted by $\mathscr{Q}_{\hbar,1}^\pm$,
of the matrix $\mathscr{Q}_{\hbar,1}$ is
\[\mathscr{Q}_{\hbar,1}^\pm=z-M_0(\sigma)-\hbar\Pi n_1 \Pi^*\,.\]
This invites to consider the effective matrix pseudo-differential operator whose symbol is
\[M_\hbar=M_0(\sigma)+\hbar M_1(s,\sigma)\,,\]
with 
\[M_1(s,\sigma)=\kappa(s)\Pi\mathscr{C}(\tau,\Xi_0(\sigma))\Pi^*\,,\quad \mathscr{C}(\tau,\xi)=(\xi-\tau) \zeta_\mu^2\tau^2+\zeta_\mu\partial_\tau+2\zeta_\mu\tau(\xi-\tau)^2\,.\]
Using again the composition theorem  to deal with the remainder $\mathscr{R}_\hbar$, we get
\begin{equation}\label{eq.invgauche}
\mathscr{Q}^{\mathrm{W}}_{\hbar,1}\mathscr{P}_\hbar^{\mathrm{W}}=\mathrm{Id}+\mathscr{O}_{L^2( \T_{2L}\times\R_+,\langle\tau\rangle^N\mathrm{d}s\mathrm{d}\tau)\times L^2( \T_{2L})\to L^2( \T_{2L}\times\R_+)\times L^2( \T_{2L})}(\hbar^2)+\mathscr{Q}^{\mathrm{W}}_{\hbar,1}\begin{pmatrix}
	\hbar w_{\hbar}\partial_\tau&0\\
	0&0
\end{pmatrix}^\mathrm{W}\,.
\end{equation}
Moreover, similar arguments show that $\mathscr{Q}^{\mathrm{W}}_{\hbar,1}$ is also an approximate right inverse of $\mathscr{P}_\hbar^{\mathrm{W}}$ in the sense that
\begin{equation}\label{eq.invdroite}
\mathscr{P}_\hbar^{\mathrm{W}}	\mathscr{Q}^{\mathrm{W}}_{\hbar,1}=\mathrm{Id}+\mathscr{O}_{L^2( \T_{2L}\times\R_+,\langle\tau\rangle^N\mathrm{d}s\mathrm{d}\tau)\times L^2( \T_{2L})\to L^2( \T_{2L}\times\R_+)\times L^2( \T_{2L})}(\hbar^2)+\begin{pmatrix}
		\hbar w_{\hbar}\partial_\tau&0\\
		0&0
	\end{pmatrix}^\mathrm{W}	\mathscr{Q}^{\mathrm{W}}_{\hbar,1}\,.
\end{equation}

\begin{proposition}\label{prop.LhtoMh}
The spectrum of ${\mathscr{L}}_h$ in $[h a,h b]$ coincides (with multiplicity) with that of  $h\mathrm{Op}^{\mathrm{W}} _\hbar M_\hbar$ modulo $\mathscr{O}(h^2)$. 
\end{proposition}
\begin{proof}
First, we consider $\psi$ an eigenfunction of $\mathscr{L}_h$  associated with $\lambda\in[h a,h b]$. We use \eqref{eq.invgauche} with $z=h^{-1}\lambda$ o get that
\[\mathscr{Q}^{\mathrm{W}}_{\hbar,1}\mathscr{P}_\hbar^{\mathrm{W}}\begin{pmatrix}
\varphi\\
0
\end{pmatrix}=\begin{pmatrix}
\varphi\\
0
\end{pmatrix}+\mathscr{O}(\hbar^2)\|\varphi\|\,,\] where $\varphi$
denotes the function $\psi$ after multiplication by a cutoff function
in $t$ and rescaling as in Proposition \ref{prop.agmonmicro}. Note
that we used the exponential decay in $\tau$ of our quasimode
$\varphi$ (which comes from that of $\psi$) to control the remainder
term in \eqref{eq.invgauche}. We infer that
\begin{equation}\label{eq.PP*}
\mathfrak{P}^*\mathfrak{P}\varphi=\varphi+\mathscr{O}(\hbar)\|\varphi\|\,,\quad (\hat\lambda-\mathrm{Op}^\mathrm{W}_\hbar M_\hbar)\mathfrak{P}\varphi=\mathscr{O}(\hbar^2)\|\varphi\|\,,
\end{equation}
where we used that the principal symbol of the top right coefficient of $\mathscr{Q}_{\hbar,1}$ is $\Pi^*$.
Since $\mathfrak{P}^*$ is bounded uniformly in $\hbar$ (as the quantization of a bounded symbol), the first relation implies that
\[\|\varphi\|\leq C\|\mathfrak{P}\varphi\|\,.\]
Then, from the second relation and the spectral theorem, we deduce that
\[\mathrm{dist}(\hat\lambda,\mathrm{sp}(\mathrm{Op}^\mathrm{W}_\hbar M_\hbar))\leq C\hbar^2\,.\]
This means that the spectrum of $h^{-1}\mathscr{L}_h$ in the window $[h a,h b]$ is at a distance of order $\hbar^2$ to the spectrum of the effective operator $\mathrm{Op}^\mathrm{W}_\hbar M_\hbar$.

Let us now proceed as in the proof of Proposition \ref{prop.firstspectralred} and keep the same notation. We have
\[\|(\oplus_{j=1}^p\mathscr{N}_\hbar-\hat\lambda)\varphi\|\leq\varepsilon_h\|\varphi\|\,,\quad \varepsilon_h=\mathscr{O}(h^\infty)\,,\]
for all $\varphi=(\varphi_1,\ldots,\varphi_p)\in \bigoplus_{j=1}^p \hat\chi_\hbar \hat E_j$ and where $\hat\lambda=(\hat\lambda_1,\ldots,\hat\lambda_p)$. 

Similarly as \eqref{eq.PP*}, we have
\begin{equation}\label{eq.injective}
\|\varphi\|\leq C\|\mathfrak{P}\varphi\|\,,
\end{equation}
and
\[(\oplus_{j=1}^p\mathrm{Op}^\mathrm{W}_\hbar M_\hbar-\hat\lambda)\mathfrak{P}\varphi=\mathscr{O}(\hbar^2)\|\mathfrak{P}\varphi\|\,,\]
where $\mathfrak{P}\varphi=(\mathfrak{P}\varphi_1,\ldots,\mathfrak{P}\varphi_p)$. Due to \eqref{eq.injective}, the action of the map $\mathfrak{P}$ is injective on $\bigoplus_{j=1}^p \hat\chi_\hbar \hat E_j$. Therefore, as in the proof of Proposition \ref{prop.firstspectralred}, the spectral theorem provides us with
\[\mathrm{rank}\,\mathds{1}_{J_h}(\mathscr{L}_h)\leq \mathrm{rank}\,\mathds{1}_{K_h}(hM^{\mathrm{W}}_\hbar)\,,\]
where $J_h\subset I_h$ and $K_h$ is an interval such that $J_h\subset K_h$ and $\mathrm{d}_{\mathrm{H}}(J_h,K_h)=\mathscr{O}(h^2)$.

Let us now prove the converse estimate. We use \eqref{eq.invdroite} with an eigenvalue $z=\hat\lambda$ of $M_\hbar^{\mathrm{W}}$ and for $f$ a corresponding eigenfunction. We have
\begin{equation}\label{eq.invdroite0f}
\mathscr{P}_\hbar^{\mathrm{W}}	\mathscr{Q}^{\mathrm{W}}_{\hbar,1}\begin{pmatrix}
0\\
f
\end{pmatrix}=\begin{pmatrix}
0\\
f
\end{pmatrix}+\mathscr{O}(\hbar^2)\|f\|\,,
\end{equation}
where the remainder term involving $w_\hbar$ has been controlled by
using the exponential decay of the eigenfunctions of the de Gennes -
Robin operator $n_0$.

Then, the first line in \eqref{eq.invdroite0f} gives
\begin{equation}\label{eq.quasimodesMhtoNh}
(\mathscr{N}_\hbar-\hat\lambda)\left(\mathscr{Q}^+_{\hbar,1}\right)^\mathrm{W} f=\mathscr{O}(\hbar^2)\|f\|\,,
\end{equation}
whereas the second line gives
\[\mathfrak{P}\left(\mathscr{Q}^+_{\hbar,1}\right)^\mathrm{W} f=\mathscr{O}(\hbar^2)\|f\|\,,\]
which leads to
\[\|f\|\leq C\left\|\left(\mathscr{Q}^+_{\hbar,1}\right)^\mathrm{W} f\right\|\,.\]
With \eqref{eq.quasimodesMhtoNh}, we get
\[(\mathscr{N}_\hbar-\hat\lambda)\left(\mathscr{Q}^+_{\hbar,1}\right)^\mathrm{W}
  f=\mathscr{O}(\hbar^2)\left\|\left(\mathscr{Q}^+_{\hbar,1}\right)^\mathrm{W}
    f\right\|\,.\] Now, by using the exponential decay of
$\left(\mathscr{Q}^+_{\hbar,1}\right)^\mathrm{W} f$ and the rough
microlocalization of $f$ in the support of $\Xi_0$ (since the
principal symbol of the scalar pseudodifferential operator $M_\hbar$
is $n_0$), we get the quasimode estimate
\[(\mathscr{L}_h-\lambda)\psi^{\mathrm{quasi}}=\mathscr{O}(h^2)\|\psi^{\mathrm{quasi}}\|\,,\]
with $\psi^{\mathrm{quasi}}(x)=\chi(t(x)/\hbar^{1-\gamma})\Psi^{\mathrm{quasi}}\circ\Gamma^{-1}(x)$
where
\[\Psi^{\mathrm{quasi}}(s,t)=\left(\mathscr{Q}^+_{\hbar,1}\right)^\mathrm{W}  f(s,\hbar^{-1}t)\,,\]
and $\chi$ is a smooth cutoff function equal to $1$ near $0$ and $0$
away from a neighborhood of $t=0$ and $\gamma\in(0,1)$ is chosen small
enough so that $t\zeta(h^{-\frac12+\eta}t)=t$ on the support of
$\chi(t/\hbar^{1-\gamma})$. Note that $\Psi^{\mathrm{quasi}}$
satisfies the Robin condition at $t=0$ since $\mathscr{Q}^+_{\hbar,1}$
(as well as $\mathscr{Q}_0$, see \eqref{eq.Qhbar,1}) takes values in a
space of functions satisfying the Robin condition. In particular,
$\psi^{\mathrm{quasi}}$ belongs to the domain of $\mathscr{L}_h$.

The spectral theorem shows that $\lambda$ is close to the spectrum of $\mathscr{L}_h$ at a distance or order at most $\mathscr{O}(h^2)$. The argument concerning the multiplicities can again be used (as above) by exchanging the roles of $\mathscr{N}_\hbar$ and $M_\hbar^{\mathrm{W}}$. The conclusion follows.
\end{proof}
\begin{remark}
  In Proposition \ref{prop.firstspectralred} we only proved one
  inclusion of spectra. In contrast, Proposition \ref{prop.LhtoMh} is
  stronger, since it provides an equality modulo $\O(h^2)$, in the
  sense of Definition~\ref{defi:coincide}. Indeed, in the proof of
  Proposition \ref{prop.LhtoMh}, we only have to use quasimodes for
  $\mathscr{N}_\hbar$ and not necessarily the true eigenfunctions of
  $\mathscr{N}_\hbar$ (whose existence in the spectral window of
  interest is not obvious). Our presentation avoids the spectral
  analysis of $\mathscr{N}_\hbar$ (existence of the discrete spectrum,
  Agmon estimates, etc.) by comparing directly the spectra of
  $\mathscr{L}_h$ and of the effective operator.
\end{remark}

\section{Analysis of the effective operator}\label{sec.5}
This section is devoted to the spectral study of
$M_\hbar^{\mathrm{W}}$ in $[a,b]$. Let us diagonalize this operator,
up to a remainder of order $\mathscr{O}(\hbar^2)$.  Note that, by
using the exponential decay of the eigenfunctions of $n_0$, we may
(and so do we) replace $\zeta_\mu$ by $1$.

\subsection{Asymptotic diagonalization and end of the proof of Theorem \ref{thm.main}}\label{sec.expA}
The end of the proof follows from classical arguments (see, for instance, \cite[Section 3.1]{HS90} where such arguments are used).
We notice that the spectrum of 
\[\mathscr{T}_\hbar=\exp(\hbar A^{\mathrm{W}}) M_\hbar^{\mathrm{W}} \exp(-\hbar A^{\mathrm{W}})\]
is the same as the one of $M_\hbar^{\mathrm{W}}$, as soon as $A$ belongs to $S(1)$ and is $2L$-periodic with respect to $s$. In this case, we recall that $A^{\mathrm{W}}$ is bounded from $L^2(\T_{2L})$ to $L^2(\R_{2L})$ (and thus its exponential is well-defined as an element of $\mathscr{L}(L^2(\T_{2L})$) thanks to the classical power series). Let us explain how to choose $A$. By expanding the exponential, we have
\[\mathscr{T}_\hbar=(\mathrm{Id}+\hbar A^{\mathrm{W}}) M_\hbar^{\mathrm{W}} (\mathrm{Id}-\hbar A^{\mathrm{W}})+\mathscr{O}(\hbar^2)\,,\]
and thus
\[\mathscr{T}_\hbar=M_\hbar^{\mathrm{W}}+\hbar[A^{\mathrm{W}},M_\hbar^{\mathrm{W}}]+\mathscr{O}(\hbar^2)\,,\]
so that
\[\mathscr{T}_\hbar=M_0^{\mathrm{W}}+\hbar\left(M_1+[A,M_0]\right)^{\mathrm{W}}+\mathscr{O}(\hbar^2)\,.\]
Therefore, $A$ should be chosen so that $M_1-[M_0,A]$ is diagonal. The map $\mathrm{Skew}_N(\R)\ni A\mapsto [M_0,A]\in\mathrm{Sym}^0_N(\R)$ is well-defined and an isomorphism since $M_0$ is diagonal with distinct real entries, where $\mathrm{Skew}_N(\R)$ is the vector space of skew-symmetric matrices and $\mathrm{Sym}^0_N(\R)$ the space of symmetric matrices with null diagonal. It is actually easy to compute its inverse. Consider $M$ a symmetric matrix with null diagonal. We want to find $A\in\mathrm{Skew}_N(\R)$ such that $[M_0,A]=M$. For all $j\in\{1,\ldots,N\}$, we have
\[(M_0-\overset{\circ}{\mu}_j)Ae_j=Me_j\,,\]
and then, for all $k\in\{1,\ldots,N\}$,
\[(\overset{\circ}{\mu}_k-\overset{\circ}{\mu}_j)\langle Ae_j,e_k\rangle=\langle Me_j,e_k\rangle\,.\]
Thus, for all $k\neq j$,
\[\langle Ae_j,e_k\rangle=(\overset{\circ}{\mu}_k-\overset{\circ}{\mu}_j)^{-1}\langle Me_j,e_k\rangle\,,\]
which determines a unique $A\in\mathrm{Skew}_N(\R)$. 

Since there is a uniform gap between the $\overset{\circ}{\mu}_j$ (with respect to $\sigma$), we get the existence of a skew-symmetric $A$ in $S(1)$ such that
\[\underbrace{M_1-\mathrm{diag}(M_1)}_{\in\mathrm{Sym}_N^0(\R)}+[A,M_0]=0\,.\]
With this choice, we get
\[\mathscr{T}_\hbar=M_0^{\mathrm{W}}+\hbar\mathrm{diag}(M_1^{\mathrm{W}})+\mathscr{O}(\hbar^2)\,.\]
Note that, for all $j\in\{1,\ldots,N\}$, we have $\langle M_1 e_j,e_j\rangle=\kappa(s)\langle \mathscr{C}(\tau,\Xi_0(\sigma))u^{[\Xi_0(\sigma),\gamma]}_j,u^{[\Xi_0(\sigma),\gamma]}_j\rangle$.
By the spectral theorem, we deduce that the spectra of $M_\hbar^{\mathrm{W}}$ and $M_0^{\mathrm{W}}+\hbar\mathrm{diag}(M_1^{\mathrm{W}})$ coincide modulo $\mathscr{O}(\hbar^2)$. This procedure can be continued at any order.

\subsection{Spectral consequences}
The aim of this last section is to prove Proposition
\ref{prop.pseudo1D} and Corollary \ref{cor.lowlying}.

\subsubsection{Proof of Proposition \ref{prop.pseudo1D}}
\label{sec:proof-BSreg}

We could not find this particular statement in the literature, because
a) we have to deal with non-connected level sets of the principal
symbol, and b) we have Floquet periodic conditions, with
$\h$-dependent Floquet exponent. The first issue is treated with usual
microlocal arguments: each connected component carries with itself a
Bohr-Sommerfeld asymptotic series, as in~\cite{HR84}, and the initial
spectrum is obtained, modulo $\O(\h^\infty)$, by the superposition
(with multiplicities) of all these series. The second one is easily
included in the general theory thanks to the ``sheaf'' approach
of~\cite{VN-Bohr-Sommerfeld00}. Indeed, near each point of the energy
level curve $\sigma=\textup{const}$, the operator $\mathsf{P}_\h$ is
microlocally a usual $\h$-pseudo-differential operator, and the
quantum Darboux-Carathéodory normal form holds. Therefore, the
Bohr-Sommerfeld cocycle of~\cite[Proposition
5.6]{VN-Bohr-Sommerfeld00} holds; the difference being that the
condition for a global section should include the Floquet exponent
$\theta$. This gives a Bohr-Sommerfeld rule for quantized energies $E$
(for each connected component) of the form
  \begin{equation}
    \mathcal{A}(E) + \h m(E)\tfrac{\pi}{2} + \h\mathcal{K}(E) +
    \O(\h^2) = 2\pi\h(\ell + \tfrac{L}{\pi}\theta)\,, \quad
    \ell\in\Z\,,
    \label{equ:BS}
  \end{equation}
  where $\mathcal{A}(E)$ is the action integral (here
  $\mathcal{A}(E)=2L\sigma$ when $E=\mu(\sigma)$), $m(E)$ the Maslov
  index (which vanishes here, because the curves
  $\sigma=\textup{const}$ project diffeomorphically on the $s$
  variable), and $\mathcal{K}(E)$ is the integral of the subprincipal
  form~\cite[Definition 3.2]{VN-Bohr-Sommerfeld00} along the energy
  level set. In order to compute $\mathcal{K}$, we notice that the
  Hamiltonian vector field of $\mu(\sigma)$ is
  $\mu'(\sigma)\tfrac{\partial}{\partial s}$ and hence the
  subprincipal form is $\frac{-r}{\mu'(\sigma)} \textup{d}s$, where
  $r$ is the subprincipal symbol of $\mathsf{P}_\h$ (here
  $r= -C(\sigma)\kappa(s)$).  Hence, for $E=\mu(\sigma)$, we have
  \[
    \mathcal{K}(E) = \frac{C(\sigma)}{\mu'(\sigma)}_{| \Sigma_q}
    \int_0^{2L}\kappa(s)\textup{d}s\,.
  \]
  Inverting the formal series~\eqref{equ:BS}, we get
  \[
    E = \mu(\sigma) -
    \mathcal{K}(\mu(\sigma))\tfrac{1}{2L}\mu'(\sigma) + \O(\h^2)\,, \quad \sigma 
 =  \tfrac{\pi}{L}\h(\ell + \tfrac{L}{\pi}\theta)  \]
  which gives~\eqref{equ:action} and~\eqref{equ:kappa}.

\subsubsection{Proof of Corollary \ref{cor.lowlying}}
\label{sec:proof-coroll-minipuits}

Thanks to Proposition \ref{prop.LhtoMh} and the considerations in Section \ref{sec.expA}, we know that the spectrum of $\mathscr{L}_h$ in $[h a,h b]$ coincides with that of $hM_\hbar^{\mathrm{W}}$ modulo $\mathscr{O}(h^2)$. In the present section, since we are interested in the low-lying eigenvalues, we take $a=-\infty$ and $b=\Theta_0(\gamma)+\varepsilon<1$ (for $\varepsilon>0$ small enough). Therefore, we have $N=1$ and the matrix symbol $M_\hbar$ reduces to a scalar symbol:
\[M_\hbar(s,\sigma)=\mu_1(\gamma,\sigma)+\hbar\kappa(s)\langle
  \mathscr{C}(\tau,\Xi_0(\sigma))u^{[\Xi_0(\sigma),\gamma]}_1,u^{[\Xi_0(\sigma),\gamma]}_1\rangle\,.\]
We are interested in the spectrum of $M_\hbar^{\mathrm{W}}$ (when
acting on $e^{is\mathfrak{f}_0/h}L^2(\T_{2L})$).  Hence, Corollary
\ref{cor.lowlying} can be obtained by~\cite{D-VN21} (see in particular
the Morse case, section 6.3.1) followed by a standard Birkhoff normal
form (here, the Floquet exponent $\mathfrak{f}_0/h$ plays no role
because the analysis is local near a point in the boundary
$\partial\Omega$). Here are the details.

Thanks to the Weyl asymptotic formula for pseudodifferential operators (see, for instance, \cite[Theorem 14.11]{Zworski}), the counting function $\mathrm{N}(M_\hbar^{\mathrm{W}},\Theta_0(\gamma)+\varepsilon)$ (giving the number of eigenvalues less than $\Theta_0(\gamma)+\varepsilon$) satisfies
\[\begin{split}\mathrm{N}(M_\hbar^{\mathrm{W}},\Theta_0(\gamma)+\varepsilon)&=\frac{1}{2\pi\hbar}\int_{\{(s,\sigma) : \mu_1(\sigma)\leq\Theta_0(\gamma)+\varepsilon\}}\mathrm{d}s\mathrm{d}\sigma+o(\hbar^{-1})\\
&=\frac{L}{\pi\hbar}|\{\sigma : \mu_1(\sigma)\leq\Theta_0(\gamma)+\varepsilon\}|(1+o(1))
	\,.\end{split}\]
Now, we take $\varepsilon=\hbar^{\eta}$, for some given $\eta>0$.

Due to the non-degeneracy of the minimum of $\sigma\mapsto\mu_1(\gamma,\Xi_0(\sigma))$, the eigenfunctions associated with eigenvalues less than $b$ are microlocalized in a neighborhood of $\xi_0(\gamma)$ of size $\hbar^{\eta/2}$ (and so are all the linear combinations of such eigenfunctions due to the Weyl estimate). This invites us to expand the symbol near $\xi_0(\gamma)$:
\begin{equation}\label{eq.TaylorMhbar}
M_\hbar(s,\sigma)=\Theta_0(\gamma)+\frac{\partial^2_\sigma\mu(\gamma,\xi_0(\gamma))}{2}(\sigma-\xi_0(\gamma))^2-\hbar\kappa(s)C_1(\xi_0(\gamma))+\mathscr{O}(|\sigma-\xi_0(\gamma)|^3+\hbar|\sigma-\xi_0(\gamma)|)\,.
\end{equation}
Therefore, $M_\hbar$ is relative perturbation of the symbol of a classical electric Schrödinger operator. The corresponding operator is 
\[\mathscr{M}_\hbar=\Theta_0(\gamma)+\frac{\partial^2_\sigma\mu(\gamma,\xi_0(\gamma))}{2}(\hbar D_s-\xi_0(\gamma))^2-\hbar\kappa(s)C_1(\xi_0(\gamma))\,.\]
Let us only consider the case when $\gamma<\gamma_0^{[0]}$ (\emph{i.e.}, $\epsilon=1$). The assumption that $\kappa$ has a unique maximum, which is non-degenerate, allows to use the harmonic approximation near the maximum of $\kappa$ (and even a Birkhoff normal form, see, for instance, \cite[Chapter 5]{Raymond17} or the original references \cite{Sj92, CVN08}). The eigenvalues of $\mathscr{M}_\hbar$ satisfy
\[\lambda_j(\mathscr{M}_\hbar)=\Theta_0(\gamma)-\kappa_{\max}C_1(\xi_0(\gamma))\hbar+\left((2j-1)\sqrt{\frac{k_2C_1(\xi_0(\gamma))\mu''_1(\gamma,\xi_0(\gamma))}{4}}\right)\hbar^{\frac32}+o(\hbar^{\frac32})\,,\]
uniformly in $j\geq 1$ such that $j\hbar^{\frac12}=o(1)$.

We recall that $\hbar=h^{\frac12}$. We get Corollary \ref{cor.lowlying} by noticing, thanks to a perturbation analysis using \eqref{eq.TaylorMhbar},  that the spectra of $hM_\hbar^{\mathrm{W}}$ and $h\mathscr{M}_\hbar$ below  $h(\Theta_0(\gamma)+\hbar^{\eta})$ coincide modulo $o(h^{\frac74})$.

\subsubsection{Proof of Theorem~\ref{theo:weyl}}
By Theorem~\ref{thm.main}, and Definition~\ref{defi:coincide} we have,
for $\epsilon=\O(h)$,
\begin{equation}
  N(h\mathfrak{M}_h, [h(a+\epsilon), h(b-\epsilon)]) \leq N(\mathscr{L}_h,[ha, hb]) \leq N(h\mathfrak{M}_h, [h(a-\epsilon), h(b+\epsilon)])\,.
  \label{equ:weyl_0}
\end{equation}
From Corollary~\ref{cor.reg}, for any interval $[a',b']$ disjoint from
$\Theta$ and $\Lambda$, the number of eigenvalues of $\mathfrak{M}_h$
in $[ha',hb']$ is bounded by $C\frac{b' - a'}{h^{1/2}}$ for some
constant $C>0$. Applying this with $(a',b')$ equal, respectively, to
the four intervals $(a,a+\epsilon)$, $(b-\epsilon,b)$,
$(a-\epsilon,a)$, and $(b,b+\epsilon)$, it follows
from~\eqref{equ:weyl_0} that
\[
  N(\mathscr{L}_h,[ha, hb]) = N(h\mathfrak{M}_h, [ha, hb]) +
  \O(\epsilon h^{-1/2})=N(h\mathfrak{M}_h, [ha, hb]) +\O(h^{1/2})\,.
\]
Therefore, it is enough to estimate $N(h\mathfrak{M}_h, [ha, hb])$, for
which we apply Corollary~\ref{cor.reg} (which is actually a
description of the spectrum of $\mathfrak{M}_h$). This corollary says
that the number of eigenvalues of $h\mathfrak{M}_h$ inside $[ha,hb]$,
including multiplicities, is given, modulo $\O(h^2)$, by the
number of integers $\ell\in\Z$ such that
\begin{equation}
  h^{\frac12}(\tfrac{\pi}{L}\ell + \theta (h)) \in
  f_{k,q,h}^{-1}([a,b])\,,\label{equ:preimage}
\end{equation}
for some admissible $(k,q)$, where
$f_{k,q,h}(\sigma) := f_{k,q}(\sigma,h^{1/2})$. 

To simplify notations, let us momentarily fix
$(k,q)$ and denote $f_h:=f_{k,q,h} =: f_0+h^{1/2}f_1+\O(h)$, where
$f_0$ and $f_1$ are defined in~\eqref{equ:fkq0} and~\eqref{equ:fkq1}.
By assumption, $f_0$ is monotonous on $\Sigma_{k,q}$, let us assume
that it is increasing; the decreasing case is obtained by swapping
$(a,b)$. For $h$ small enough, $f_h$ is also increasing and hence
$f_{h}^{-1}([a,b]) = [f_h^{-1}(a), f_h^{-1}(b)]$. Therefore, the
solutions to~\eqref{equ:preimage} are exactly the integers belonging
to the interval
\begin{equation}
  \frac{Lh^{-1/2}}{\pi}[f_h^{-1}(a), f_h^{-1}(b)]-\frac{L}{\pi}\theta(h)\,.
  \label{equ:interval}
\end{equation}
Let $\sigma = f_h^{-1}(a)$; of
course $\sigma$ depends on $h$, but since $\sigma\in\Sigma_{k,q}$, it
is bounded and we have $\sigma = f_0^{-1}(a) +
\O(h^{1/2})$. Therefore,
$f_1(\sigma) = f_1(f_0^{-1}(a)) + \O(h^{1/2})$. According to the
statement of Theorem~\ref{theo:weyl}, we denote $\alpha:=f_0^{-1}(a)$. Writing $f_0(\sigma) = a - h^{1/2}f_1(\alpha) + \O(h)$ we get, by Taylor expansion,
\[
\sigma = \alpha - h^{1/2}(f_0^{-1})'(a) f_1(\alpha) + \O(h) = \alpha + h^{1/2}\frac{ \langle \kappa \rangle C_k(\alpha)}{\mu_k'(\alpha)} + \O(h)\,.
\]
Using the analogous formula for $f_h^{-1}(b)$, we may compute the
difference $f_h^{-1}(b) - f_h^{-1}(a)$ and obtain the length of the
interval~\eqref{equ:interval}:
\[
  \frac{Lh^{-1/2}}{\pi}(f_h^{-1}(b) - f_h^{-1}(a)) =
  \frac{Lh^{-1/2}}{\pi}(\beta - \alpha) + \frac{L\langle \kappa
    \rangle}{\pi} \left(\frac{C_k(\beta)}{\mu_k'(\beta)} -
  \frac{C_k(\alpha)}{\mu_k'(\alpha)}\right) + \O(h^{1/2})
\]
which gives Theorem~\ref{theo:weyl} by summing over admissible
$(k,q)$.

\subsubsection{Proof of Theorem~\ref{theo:oscillate}}
\label{sec:proof-theor-oscillate}

We use the notation of Theorem~\ref{thm.main}. By
Proposition~\ref{prop.pseudo1D}, the self-adjoint operators $m_k^W$
acting on $e^{i\theta(h)\cdot}L^2(\T_{2L})$ satisfy the Gårding
inequality:
\[
  m_k^W \geq \min \mu_k - \O(h^{1/2}) = \Theta^{[k-1]} - \O(h^{1/2}) >
  b \qquad \forall k=2,\dots,N, \quad \forall h < h_0
\]
for $h_0$ small enough. Hence the spectrum of $\mathscr{L}_h$ in $I_h$
coincides, modulo $\O(h^2)$, with the spectrum of $h m_1^W$ in that
interval. In other words, for this choice of interval $I_h$, the
disjoint unions of Corollary~\ref{cor.reg} reduce to a union of the
two components $(k=1,q=1)$ and $(k=1,q=2)$, and the spectrum in $I_h$
coincides modulo $\O(h^2)$ with
 \[
   \bigsqcup_{q=1,2} \left\{ h f_{1,q} (\sigma,h^{\frac12}), \, \sigma
     \in h^{\frac12}(\tfrac{\pi}{L}\Z + \theta (h)) \cap \Sigma_{1,q}
   \right\}  \cap [ha, hb]\,.
\]
So eigenvalues $\lambda_j$ in $I_h$ are associated with integers
$\ell = \ell(h)\in\Z$ such that
$ h^{\frac12}(\tfrac{\pi}{L}\ell + \theta (h)) \in \Sigma_{1,1} \cup
\Sigma_{1,2}$; therefore there are constants $\alpha,\beta$,
independent on $h$, such that
\[
  \sigma_\ell(h) := h^{\frac12}(\tfrac{\pi}{L}\ell + \theta (h)) \in
  [\alpha,\beta]\,.
\]
Hence
$\tfrac{\pi}{L}\ell \in [\frac{\alpha}{h^{1/2}} - \theta(h),
\frac{\beta}{h^{1/2}} - \theta(h)] $. Recalling that
$\theta(h)=\frac{|\Omega|}{|\partial\Omega|h}$, we get that, for
$h< h_0:=\sqrt{\frac{|\Omega|}{\beta|\partial\Omega|}}$, $\ell$ must
be negative. Thus, for each fixed $\ell$, $\sigma_\ell(h)$ increases
when $h$ decreases to zero. In other words, because of the non-zero
flux term, the corresponding semiclassical eigenvalues
$ h f_{1,q} (\sigma_\ell(h),h^{\frac12})$ ``move to the right'' (in
the sense of Figure~\ref{fig:dispersion}) towards the Landau level
$\mu_1=1$.

Let us now describe the semiclassical branches, \emph{i.e.} the curves
\begin{equation}
  h \mapsto  h f_{1,q} (\sigma_\ell(h),h^{\frac12}), \qquad q=1,2, \quad \ell\in\Z_-.
  \label{equ:branches}
\end{equation}
We may assume that the intervals $\Sigma_{1,q}$ satisfy
$\Sigma_{1,1} \leq \Sigma_{1,2}$. Recall from
Proposition~\ref{prop.dispersion} (or Figure~\ref{fig:dispersion})
that there exists $c>0$ such that ${\mu_1'}_{| \Sigma_{1,1}} \leq - c$
while ${\mu_1'}_{| \Sigma_{1,2}} \geq c$. Hence, in view of the
semiclassical expansion of $f_{k,q}$, and up to reducing $c$, we get
\begin{equation}
f_{1,1}' \leq - c \qquad \text{ and } \qquad
f_{1,2}' \geq c
\label{equ:f_prime}
\end{equation}
uniformly for $h\leq h_0$ small enough. Thus, for each fixed
admissible $\ell\in\Z_-$, the branch~\eqref{equ:branches} generated by
$\Sigma_{1,1}$ (\emph{i.e.} corresponding to $q=1$) is an increasing
curve, while the branch corresponding to $q=2$ is decreasing.
Moreover, the semiclassical branches generated by $\Sigma_{1,1}$ and
associated with different integers $\ell_1\neq \ell_2$ will never
cross as $h$ varies, and their mutual vertical distance is bounded
below as
\begin{equation}
\abs{h f_{1,q} (\sigma_{\ell_1}(h),h^{\frac12}) - h f_{1,q} (\sigma_{\ell_2}(h),h^{\frac12})
} \geq h^{\frac{3}{2}}\frac{\pi c}{L}\,.
\label{equ:vertical_distance}
\end{equation}
Hence, in view of~\eqref{equ:f_prime}, we see that the horizontal
distance between these curves is $\O(h^2)$.  Of course, the same holds
for the branches associated with $\Sigma_{1,2}$, which thus form a
collection of disjoint decreasing curves. Therefore, the superposition
of all branches is a deformed grid intersected with the window
$(0,h_0]\times [ha, hb]$, see Figure~\ref{fig:branches}. In
particular, there are many crossing points, and the horizontal
distance between consecutive crossing points along a fixed branch is
$\O(h^2)$.
\begin{figure}[h]
  \centering
  \includegraphics[width=0.7\linewidth]{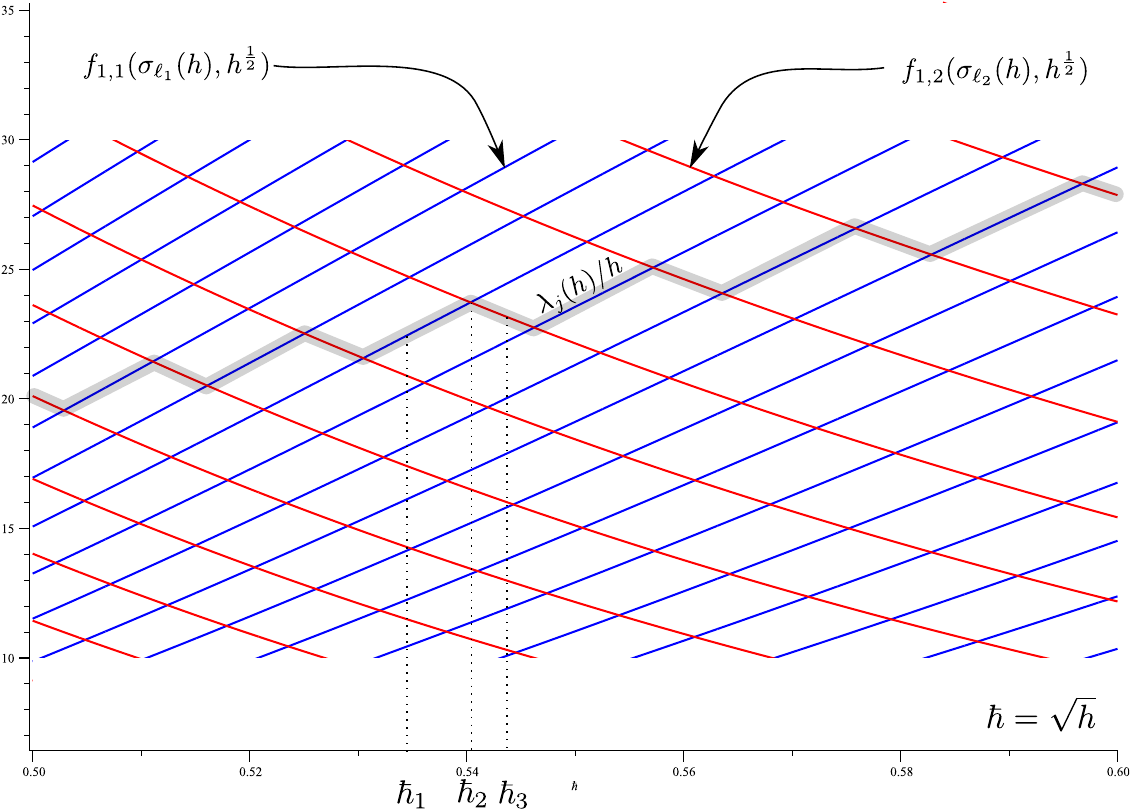}
  \caption[Semiclassical branches of eigenvalues]{Illustration of the
    collection of semiclassical branches of eigenvalues. Here we plot
    the graphs of $f_{1,q}(\sigma_\ell(h),h)$ with respect to the
    variable $\hbar=\sqrt{h}$, for $q=1$ (blue curves) and $q=2$ (red
    curves). The continuous curve of $\lambda_j(h)/h$, for fixed $j$,
    where $\lambda_j(h)$ is the exact eigenvalue of $\mathscr{L}_h$,
    lies within the greyed stair-case like curve (of vertical width
    $\O(h)$).}
  \label{fig:branches}
\end{figure}

Consider now the exact eigenvalues $\lambda_j\in [ha,hb]$. By
Corollary~\ref{cor.reg}, each $\lambda_j$ must be $\O(h^2)$-close
to one of the semiclassical branches. For fixed $\ell\in\Z_-$,
modifying the value of $h$ by an amount of order $\O(h^2)$ amounts to
shifting the abscissa $\sigma_\ell(h)$ by an amount proportional to
$h^{\frac12}$. Therefore, by suitably choosing $C_1$ and setting
$h_1:= h + C_1 h^2$ we may assume that $\lambda_j(h_1)$ corresponds to a
unique increasing branch (parameterized by $\Sigma_{1,1}$): when $h$
varies in an interval of size $\O(\varepsilon h^2)$ around $h_1$, with
$\varepsilon>0$ small enough, there is a unique and fixed
$\ell_1\in\Z$ such that
\[
  \abs{\lambda_j(h) - h f_{1,1} (\sigma_{\ell_1}(h),h^{\frac12})} =
  \O(h^2)\,.
\]
Next, we choose $C_2>C_1$ so that with $h_2:= h + C_2h^2$,
$\sigma_{\ell_1}(h_2)$ is $\O(h^3)$ close to the first crossing on the
right hand side of $\sigma_{\ell_1}(h_1)$. The exponent $3$ is not
important, any exponent $N\geq 3$ will work as well.  We have
\[
  \abs{\lambda_j(h_2) - h_2 f_{1,1}
    (\sigma_{\ell_1}(h_2),h_2^{\frac12})} = \O(h^{2})\,.\]
and since $\abs{f'_{1,1}}\geq c$, we obtain a constant $C>0$ such that
\[
  \lambda_j(h_2) \geq \lambda_j(h_1) + C h^{3/2}\,.
\]

On the right hand side of the crossing, the integer $\ell_1$, and the
increasing branch, do not longer correspond to the eigenvalue
$\lambda_j$ (this branch will now correspond to
$\lambda_{j+1}$). Instead, we have to select the branch parameterized
by $\Sigma_{1,2}$, labelled by some $\ell_2\in\Z_-$; then, as before,
with a suitable $C_3>C_2$, we have, with $h_3:= h + C_3 h^2$
\[
  \abs{\lambda_j(h_3) - h_3 f_{1,2} (\sigma_{\ell_2}(h_3),h_3^{\frac12})} =
  \O(h^2)\,,
\]
and hence, since the new branch is now decreasing,
\[
\lambda_j(h_3) \leq \lambda_j(h_2) - C h^{3/2}\,.
\]

Note that, in the above analysis, the constants $C_j$ depend on $h$,
but in a uniform way: they belong to a fixed compact interval
contained in $(0,+\infty)$. The above estimates are then uniform for
$h\leq h_0$ if $h_0$ is chosen small enough.

We now turn to the last statement of the theorem. We choose $h_2$ as
before, but with more precision: we can always select the exact
crossing point $h'$ between the semiclassical branches, \emph{i.e.}:
\[
  h' f_{1,1} (\sigma_{\ell_1}(h'),h'^{\frac12}) = h' f_{1,2}
  (\sigma_{\ell_2}(h'),h'^{\frac12})\,.
\]
This gives
\[
  \lambda_j(h') - \lambda_{j+1}(h') = \O(h^2)\,.
\]
Finally, for any value of $h$ sufficiently far from the crossing, for
instance $h''=h_1$ or $h_3$, the vertical
estimate~\eqref{equ:vertical_distance} ensures that
\[
  \lambda_{j+1}(h'') - \lambda_{j}(h'') \geq Ch^{3/2}\,,
\]
for some $C>0$, which finishes the proof of the theorem.
  
\subsection*{Acknowledgments}
The authors are deeply grateful to the CIRM, where this work was initiated. They also thank Philippe Briet and Ayman Kachmar for useful discussions. This work was conducted within the France 2030 framework programme, the Centre Henri Lebesgue  ANR-11-LABX-0020-01.

\appendix
\section{A rough Weyl estimate}\label{sec.roughWeyl}\label{sec.app}
The aim of this section is recall why \eqref{eq.roughNLhIh} holds. Thanks to the Young inequality, we have, for all $\psi\in H^1(\Omega)\,,$
	\begin{equation*}
		\mathcal{Q}_{h,\mathbf{A}}(\psi)
		\geq \dfrac{h^{2}}{2} \|\nabla \psi\|^{2}-2\|A\|_{\infty}^{2}\|\psi\|^{2}+\gamma\, h^{\frac32}\int_{\partial\Omega}|\psi|^2\mathrm{d}s\,.
	\end{equation*}
When $\gamma\geq 0$, we get that
\[\mathcal{Q}_{h,\mathbf{A}}(\psi)
\geq \dfrac{h^{2}}{2} \|\nabla \psi\|^{2}-2\|A\|_{\infty}^{2}\|\psi\|^{2}\,.\]
When $\gamma<0$, we use a classical trace theorem: there exists $C>0$ such that, for all $\varepsilon >0$,
\[\int_{\partial\Omega}|\psi|^2\mathrm{d}s \leq \varepsilon \|\nabla \psi\|^{2} +  C\varepsilon^{-1} \|\psi\|^{2} \,.\]
By choosing $\varepsilon=-\frac{\sqrt{h}}{4\gamma}>0$, we deduce that
\[\mathcal{Q}_{h,\mathbf{A}}(\psi)\geq \dfrac{h^{2}}{4} \|\nabla \psi\|^{2}-2\|A\|_{\infty}^{2}\|\psi\|^{2}-4\gamma^2C h  \|\psi\|^{2}\,.\]
In both cases, there exists $\Tilde{C}>0$ such that, for all $h\in(0,1)$ and all $\psi\in H^1(\Omega)$,
\[ \mathcal{Q}_{h,\mathbf{A}}(\psi) \geq \frac{h^{2}}{4} \|\nabla \psi\|^{2}-\tilde{C} \|\psi\|^{2}\,.\]
With the min-max principle, this shows that, for all $\lambda$,
\[
	N(\mathscr{L}_h,\lambda)  \leq N\left(- \Delta^{\mathrm{Neu}},4\, \frac{\lambda+\tilde{C}}{h^{2}}\right)\,.
\]
The conclusion follows from the Weyl asymptotics for the Neumann
Laplacian, which is the same at the main order as in the Dirichlet
case, see, for instance, \cite[Introduction]{NS05}.

\section{Spectral analysis of De Gennes operator}\label{sec.deGennes}

\begin{lemma}\label{lem.lbmun}
	For each  $\gamma \in \R, n\geq2, $ we have \[ \mu_{n}(\gamma,\sigma)>2n-3\,.\]
	In particular, we have \[\Theta^{[n-1]}(\gamma)>2n-3\,.\]
\end{lemma}
\begin{proof}
	From the Sturm–Liouville theory,  $u_{n}^{[\gamma,\sigma]}$ admits $n-1$ zeros on $\R_{+}\,.$ We denote by $z_{n,1}(\gamma,\sigma) $ its first zero. We consider the function 
	\[v _{n}^{[\gamma,\sigma]}(t)=u_{n}^{[\gamma,\sigma]} ( t+ z_{n,1}(\gamma,\sigma))\,,\]
	which satisfies the Dirichlet condition at $0$ and 
	\[H^{\mathrm{Dir}}[\sigma-z_{n,1}(\gamma,\sigma)] v
          _{n}^{[\gamma,\sigma]}= \mu_{n}(\gamma,\sigma)v
          _{n}^{[\gamma,\sigma]}\,, \] where
        $H^{\mathrm{Dir}}[\sigma]$ is the Dirichlet realization of
        $-\partial_\tau^2+(\sigma-\tau)^2$ on $L^2(\R_+)$.  The
        function $v _{n}^{[\gamma,\sigma]}$ has exactly $n-2$ zeros on
        $\R_{+}\,.$ By the Sturm's oscillation theorem,
        $v_{n}^{[\gamma,\sigma]}$ is the $(n-1)$-th eigenfunction of
        $H^{\mathrm{Dir}}[\sigma-z_{n,1}(\gamma,\sigma)]\,.$ Therefore
        we have
	\[\mu_{n}(\gamma,\sigma)= \mu_{n-1}^{\mathrm{Dir}}(\sigma-z_{n,1}(\gamma,\sigma))\,.\]
	Moreover, by monotonicity of the Dirichlet problem, for all $\sigma \in \R$, \[\mu_{n-1}^{\mathrm{Dir}}(\sigma) > 2n-3\,. \]
\end{proof}

The following proposition is obtained by adapting the proof of \cite[Theorem II.2]{K06}.
\begin{proposition}
	\label{Prop.cri}
	Let $n\geq 1$. If $\sigma$	is a critical point of $\mu_n(\gamma,\cdot)$, we have
	\[\mu_n(\gamma,\sigma)=\sigma^2-\gamma^2\,.\]
\end{proposition}

\begin{lemma}
	\label{prop.Mom}
	When $\gamma \in \mathbf{R}$, we have the following relations
	\begin{equation}
		\label{Mom.1}
		\int_{0}^{+\infty} (t-\xi_{n-1}(\gamma)) |u^{[\gamma,\xi_{n-1}(\gamma)]}_n(t)|^2    \dd t=0\,,
	\end{equation}
	\begin{equation}
		\label{Mom.3}
		\int_{0}^{+\infty} (t-\xi_{n-1}(\gamma))^3 |u^{[\gamma,\xi_{n-1}(\gamma)]}_n(t)|^2    \dd t=\frac{1}{6}[1+2\gamma\xi_{n-1}(\gamma)] u^{[\gamma,\xi_{n-1}(\gamma)]}_n(0)\,.
	\end{equation}
\end{lemma}
\begin{proof}
	We let 
	\[u^{[\gamma]}_n=u^{[\gamma,\xi_{n-1}(\gamma)]}_n\,.\]
	Let us consider the differential operator:
	\[
	L=-\partial_t^2+(t-\xi_{n-1}(\gamma))^2-\Theta^{[n-1]}(\gamma)\,.
	\]
	Note that for any polynomial $p$, we have:
	\begin{equation}
		\label{ident.1}
		Lv=\left(p^{(3)}-4\left[(t-\xi_{n-1} (\gamma)^2-\Theta(\gamma)^{[n-1]} \right]p^{\prime}-4(t-\xi_{n-1}(\gamma)) p\right) u^{[\gamma]}_n, 
	\end{equation}
	and
	\begin{equation}
		\label{ident.2}
		\int_0^{+\infty} u^{[\gamma]}_n(t)(L v)(t) \mathrm{d} t=\int_0^{+\infty} L u^{[\gamma]}_n(t) v(t) \mathrm{d} t+ \left(v^{\prime}(0)-\gamma v(0)\right) u^{[\gamma]}_n(0) ,
	\end{equation}
	for  $v=2 p [u^{[\gamma]}_n]^{\prime}-p^{\prime}  u^{[\gamma]}_n.$ Taking $p=1$, we get
	\[
	-4 \int_0^{+\infty}(t-\xi_{n-1}(\gamma))\left|u^{[\gamma]}_n(t)\right|^2 \mathrm{d} t=2\left(\xi_{n-1}(\gamma)^2-\gamma^2-\Theta^{[n-1]}(\gamma)\right)\left|u^{[\gamma]}_n(0)\right|^2\,.
	\]
	Recalling Proposition \ref{Prop.cri}, the above formula proves \eqref{Mom.1}.
	To prove \eqref{Mom.3}, we take $p=(t-\xi_{n-1}(\gamma))^2$. Then, we have
	\[
	v^{\prime}(0)-\gamma v(0)=-2\left(2\gamma \xi_{n-1}(\gamma) +1\right)u^{[\gamma]}_n(0)\,. 
	\]
	We get now from \eqref{ident.1} and \eqref{ident.2}
	\[
	-12 \int_0^{+\infty}(t-\xi_{n-1}(\gamma))^3\left|u^{[\gamma]}_n(t)\right|^2 \mathrm{~d} t=-2\left(2\gamma \xi_{n-1}(\gamma) +1\right) |u^{[\gamma]}_n(0)|^2 \,.
	\]
\end{proof}

\begin{lemma}\label{eq.Cjxi}
	We have 
	\[\begin{split}
		C_j(\xi_{j-1}(\gamma))&=\frac12 \left(u_j^{[\gamma]}(0)\right)^2-\int_0^{+\infty}(t-\xi_{j-1}(\gamma))^3\left(u^{[\gamma]}_j(t)\right)^2\mathrm{d}t\\
		&=\frac13(1-\gamma\xi_{j-1}(\gamma))\left(u_j^{[\gamma]}(0)\right)^2\,,
	\end{split}\]
	where $C_j$ is defined in \eqref{eq.Cj}.
\end{lemma}
\begin{proof}
	We write
	\[(\sigma-t) t^2+2t(\sigma-t)^2=(t-\sigma)^3-\sigma^2(t-\sigma)\,.\]
	We take $\sigma=\xi_{j-1}(\gamma)$. The conclusion follows from Lemma \ref{prop.Mom}.
\end{proof}
\begin{proposition}\label{prop.sgnCj}
	Let us fix $j\geq1.$ When $\gamma \in \R$, there exists $\gamma_{0}^{[j-1]}>0$, such that, $ C_j(\xi_{j-1}(\gamma))$ is positive if $\gamma< \gamma_{0}^{[j-1]}$ and negative if  $\gamma> \gamma_{0}^{[j-1]}$.
\end{proposition}
\begin{proof}
	We notice that, for $\gamma\leq0$, we get $C_j(\xi_{j-1}(\gamma))<0.$
	
	Now, for $\gamma>0.$ From Proposition \ref{Prop.cri}, we can rewrite $C_j(\xi_{j-1}(\gamma))$ as
	\[  C_j(\xi_{j-1}(\gamma))=\frac13\left(1-\gamma \sqrt{\gamma^2+\Theta^{[j-1]}(\gamma)}\right)\left(u_j^{[\gamma]}(0)\right)^2\,,   \]
	Since $\left(u_j^{[\gamma]}(0)\right)^2>0 $, then, to study the sign of $C_j(\xi_{j-1}(\gamma))$ it is sufficient to study the sign of the function $f$ defined by $f(\gamma)=1-\gamma \sqrt{\gamma^2+\Theta^{[j-1]}(\gamma)}.$ We have 
	\[ f^\prime(\gamma)=-\sqrt{\gamma^2+\Theta^{[j-1]}(\gamma)}-\frac\gamma2  \dfrac{\left(\Theta^{[j-1]}(\gamma)\right)^\prime +2\gamma }{\sqrt{\gamma^2+\Theta^{[j-1]}(\gamma)}} \,. \]
	We can use \cite[Section B]{K06} (which can be adapted to $j\geq1$) to deduce that $ f^\prime(\gamma)<0.$ Therefore, $f$ is increasing on $[0,+\infty[.$
	
	Let us notice now that $f(0)=1$ and $\lim_{\gamma \to +\infty} f(\gamma)=-\infty.$ This establishes the existence of a unique zero of $f(\gamma)$, denoted by $\gamma_{0}^{[j-1]}$.
\end{proof}

\begin{lemma}\label{lem.sandwich}
	Let $\alpha\in\R$ and $\beta\in\N$. Consider the interval $[a,b]$. We consider $\Pi\psi=(\langle\psi, u_{j}^{[\gamma,\sigma]}\rangle)_{1\leq j\leq n}$, where $n$ is the number of dispersion curves $\mu_{j}(\gamma,\sigma)$ taking values in $[a,b]$ (see the discussion at the beginning of Section \ref{sec.K}). We consider $\hat K$ a neighborhood of $\tilde K$.

	There exists $C_{\alpha,\beta}>0$ such that for all $z\in[a,b]$ and all $\sigma\in \hat K$, the following holds. For all $v\in L^2(\R_+)$ such that $\langle t\rangle^\alpha v\in L^2(\R_+)$, we have
	\[\|\langle t\rangle^{-\alpha}\partial_\sigma^\beta(H[\gamma,\sigma]-z)^{-1}(\Pi^*\Pi)^{\perp}(\langle t\rangle^{\alpha}v)\|\leq C_{\alpha,\beta}\|v\|\,.\]

\end{lemma}
\begin{proof}
	Let us only prove this estimate for $\beta=0$.
	We consider $z\in[a,b]$. Let us consider $v\in \mathscr{S}(\overline{\R}_+)$ and let $u$ be the unique solution to the equation
	\begin{equation}\label{eq.H-z}
		(H[\gamma,\sigma]-z)u=(\Pi^*\Pi)^{\perp}(\langle \epsilon t\rangle^\alpha_{k}v)
	\end{equation}
	that is orthogonal to $(u_{j}^{[\gamma,\sigma]})_{1\leq j\leq n}$, with
	\[\langle t\rangle_k=\left(1+t^2\chi^2_k\right)^{\frac12}\,,\]
	where $\chi_k$ is a smooth non-negative function equal to $0$ on $[0,1]$ and to $1$ on $[2k,+\infty)$ and such that $|\chi'_k|\leq k^{-1}$. In particular, the weight is $1$ near $0$. Here $\epsilon>0$ is a parameter to be chosen small enough.

	We have seen in Lemma \ref{lem.bij} that 
	\[\begin{pmatrix}
		H[\gamma,\sigma]-z&\Pi^*\\
		\Pi&0
	\end{pmatrix} : B^2(\R_+)\times \C^n\longrightarrow L^2(\R_+)\times\C^n\]
	is bijective. Thus, the equation \eqref{eq.H-z} is equivalent to
	\[\begin{pmatrix}
		H[\gamma,\sigma]-z&\Pi^*\\
		\Pi&0
	\end{pmatrix}\begin{pmatrix}
		u\\
		0
	\end{pmatrix} 
	=\begin{pmatrix}
		(\Pi^*\Pi)^{\perp}(\langle \epsilon t\rangle^\alpha_{k}v)\\
		0
	\end{pmatrix}\,.\]
	Note that
	\[\begin{pmatrix}
		\langle \epsilon t\rangle^{-\alpha}_{k}	&0\\
		0&1
	\end{pmatrix}\begin{pmatrix}
		H[\gamma,\sigma]-z&\Pi^*\\
		\Pi&0
	\end{pmatrix}\begin{pmatrix}
		\langle \epsilon t\rangle^{\alpha}_{k}&0\\
		0&1
	\end{pmatrix}=\begin{pmatrix}
		H[\gamma,\sigma]+R_{\epsilon, k}-z&\langle \epsilon t\rangle^{-\alpha}_{k}\Pi^*\\
		\Pi(\cdot \langle \epsilon t\rangle^{\alpha}_{k})&0
	\end{pmatrix}\,,\]
	where
	\[R_{\epsilon, k}=-\alpha\epsilon^2(t^2\chi^2_k)'\frac{1}{1+\epsilon^2 t^2\chi^2_k}\partial_t-\left(\frac{\epsilon^2\alpha}{2}\frac{(t^2\chi^2_k)''}{1+\epsilon^2t^2\chi_k^2}+\left(\frac\alpha2-1\right)\frac{\epsilon^4\alpha}{2}\frac{(t^2\chi_k^2)'^2}{(1+\epsilon^2 t^2\chi_k^2)^2}\right)\,.\]
	With
	\[u=\langle \epsilon t\rangle^{\alpha}_{k}\tilde u\,,\]
	we get
	\[\mathscr{H}[\gamma,\sigma,\epsilon, k]\begin{pmatrix}
		\tilde u\\
		0
	\end{pmatrix} 
	=\begin{pmatrix}
		\langle \epsilon t\rangle^{-\alpha}_{k}	 (\Pi^*\Pi)^{\perp}(\langle \epsilon t\rangle^\alpha_{k}v)\\
		0
	\end{pmatrix}\,,\]
	with
	\[
	\mathscr{H}[\gamma,\sigma,\epsilon, k]=\begin{pmatrix}
		H[\gamma,\sigma]+R_{\epsilon, k}-z&\langle \epsilon t\rangle^{-\alpha}_{k}\Pi^*\\
		\Pi(\cdot \langle \epsilon t\rangle^{\alpha}_{k})&0
              \end{pmatrix}\,.\] Thanks to the exponential decay of
            the $u_j^{[\gamma,\sigma]}$ (which is uniform for
            $\sigma\in \hat K$), we notice that
            $\mathscr{H}[\gamma,\sigma,\epsilon, k]$ is bijective as
            soon as $\epsilon$ is small enough and $k$ large
            enough. Moreover,
	\[\|\mathscr{H}[\gamma,\sigma,\epsilon, k]^{-1}\|\leq C\,.\]
	This implies that
	\[\|\tilde u\|\leq C\|\langle \epsilon t\rangle^{-\alpha}_{k}	 (\Pi^*\Pi)^{\perp}(\langle \epsilon t\rangle^\alpha_{k}v)\|\,,\]
	and then (by using again the exponential decay of the eigenfunctions)
	\[\| \langle \epsilon t\rangle^{-\alpha}_{k}u\|\leq C\|v\|\,.\]
	Taking the limit $k\to+\infty$,  the Fatou lemma gives
	\[\| \langle \epsilon t\rangle^{-\alpha}u\|\leq C\|v\|\,.\]
	This provides us with the desired estimate since $\langle\epsilon t\rangle\langle t\rangle^{-1}\in[\epsilon,1]$.
\end{proof}

\bibliographystyle{abbrv}
\bibliography{FLTRVN}
\end{document}